\documentclass{ejmMod}

\usepackage[english]{babel}
\usepackage{graphicx}
\usepackage{amsmath,amssymb, amsbsy, amsgen}
\usepackage{upmath}
\usepackage[ansinew]{inputenc}
\usepackage{listings}
\usepackage{color}
\usepackage[normalem]{ulem}

\usepackage{rotating}
\usepackage{hhline}
\usepackage{multirow}
\usepackage{enumerate}
\usepackage[bookmarks]{}
\usepackage{amsfonts}   % if you want the fonts
\usepackage{tikz}
\usetikzlibrary{calc,through,backgrounds,patterns}
\usepackage{bbm}
\usepackage{pinlabel}
\usepackage{dsfont}
\allowdisplaybreaks

\newtheorem{thm}{Theorem}[section]
\newtheorem{prop}[thm]{Proposition}
\newtheorem{lem}[thm]{Lemma}

\newtheorem{cor}[thm]{Corollary}
\newtheorem{defn}[thm]{Definition}

\newdefinition{rem}[thm]{Remark}

\newcommand{\norm}[2]{ \| #1 \|_{ #2 } }
\newcommand{\haa}[1]{( #1 )}

\newcommand{\acc}[1]{ \{ #1 \} }
\newcommand{\accv}[2]{ \{ #1 \, | \, #2 \} }

\newcommand{\bighaa}[1]{\big( #1 \big)}

\newcommand{\bigabs}[1]{\big| #1 \big|}

\newcommand{\bigaccv}[2]{\big\{ #1 \, \big| \, #2 \big\}}

\newcommand{\Bighaa}[1]{\Big( #1 \Big)}
\newcommand{\Bigaccv}[2]{\Big\{ #1 \, \Big| \, #2 \Big\}}

\newcommand{\biggnorm}[2]{\bigg\| #1 \bigg\|_{ #2 } }
\newcommand{\bigghaa}[1]{\bigg( #1 \bigg)}

\newcommand{\lrhaa}[1]{\left( #1 \right)}

\newcommand{\lrabs}[1]{\left| #1 \right|}
\newcommand{\lrnorm}[2]{\left\| #1 \right\|_{ #2 }}

\newcommand{\ga}[2]{\frac{d #1 }{d #2 }}

\newcommand{\R}[1]{\mathbb{R}^{ #1 }}
\newcommand{\N}[1]{\mathbb{N}^{ #1 }}

\newcommand{\I}[1]{ \{ 1, \ldots, #1 \} }
\newcommand{\intabx}[4]{ \int_{ #1 }^{ #2 } { #3 } \, d { #4 } }

\newcommand{\intabxnohs}[4]{ \int_{ #1 }^{ #2 } { #3 } d { #4 } }

\newcommand{\ddollar}[1]{$\displaystyle{ #1 }$} % enumerate formula

\newcommand{\lrard}{\, :\Leftrightarrow \,}
\newcommand{\hsf}{\quad} % hspace formula
\newcommand{\hsone}{\,} % 1mm spacing
\newcommand{\hsq}{\;} % hspace quantors

\newcommand{\ahat}{\hat{\alpha}}

\newcommand{\anq}[1]{\alpha_n^{({ #1 })}}

\newcommand{\EF}[2]{E^{ ({ #1 }, { #2 }; \operatorname{F} ) }}
\newcommand{\EFintro}{E^{ ( \operatorname{F} ) }}
\newcommand{\Eh}{\hat{E}}
\newcommand{\EhF}{\hat{E}^{ ( \operatorname{F} ) }}
\newcommand{\Ehita}{\hat{E}^{ ( \operatorname{i} ) }}
\newcommand{\EhL}{\hat{E}^{ ( \operatorname{L} ) }}
\newcommand{\Eita}[2]{E^{ ({ #1 }, { #2 }; \operatorname{i} ) }}
\newcommand{\Eitaintro}{E^{ ( \operatorname{i} ) }}
\newcommand{\EL}[2]{E^{ ({ #1 }, { #2 }; \operatorname{L} ) }}
\newcommand{\ELintro}{E^{ ( \operatorname{L} ) }}
\newcommand{\EnF}[2]{E_n^{ ({ #1 }, { #2 }; \operatorname{F} ) }}
\newcommand{\EnFintro}{\mathcal{E}^{ ( \operatorname{F} ) }}

\newcommand{\Enh}{\hat{E}_n}
\newcommand{\EnhF}{\hat{E}_n^{ ( \operatorname{F} ) }}
\newcommand{\Enhita}{\hat{E}_n^{ ( \operatorname{i} ) }}
\newcommand{\EnhL}{\hat{E}_n^{ ( \operatorname{L} ) }}
\newcommand{\Enita}[2]{E_n^{ ({ #1 }, { #2 }; \operatorname{i} ) }}
\newcommand{\Enitaintro}{\mathcal{E}^{ ( \operatorname{i} ) }}

\newcommand{\EnL}[2]{E_n^{ ({ #1 }, { #2 }; \operatorname{L} ) }}
\newcommand{\EnLintro}{\mathcal{E}^{ ( \operatorname{L} ) }}

\newcommand{\eps}{\varepsilon}

\newcommand{\Gto}{\xrightarrow{\Gamma}}

\newcommand{\indicator}[1]{ \mathds{1}_{ \{ #1 \} } }

\newcommand{\limf}{\liminf_{ n \rightarrow \infty }}
\newcommand{\limLolone}{\beta}
\newcommand{\limp}{\limsup_{ n \rightarrow \infty }}
\newcommand{\logE}{\mathfrak{E}}
\newcommand{\logEn}{\mathfrak{E}_n}
\newcommand{\Lol}{\gamma}
\newcommand{\maxmu}{M}

\newcommand{\supp}{\operatorname{supp}}

\newcommand\Bstrut{\rule[-3ex]{0pt}{0pt}}

\newcommand{\dssetaba}{\mathcal{Y}} %Dense SubSET for cases i = 1,2, the 1-st appearance. 12345 = abcde
\newcommand{\dssetabb}{\mathcal{Y}_1} %Dense SubSET for cases i = 1,2, the 1-st appearance
\newcommand{\dssetca}{Y^{(3)}}
\newcommand{\dssetcb}{Y^{(3)}_1}

\newcommand{\dssetda}{Y^{(4)}}
\newcommand{\dssetdb}{Y^{(4)}_1}
\newcommand{\dssetdblol}{Y^{(4)}_{\Lol }}

\newcommand{\dssetea}{Y^{(5)}}
\newcommand{\dsseteaold}{\tilde{Y}^{(5)}}

\newcommand{\dsseteblol}{Y^{(5)}_{\Lol }}

\newcommand{\dsseteclol}{Z^{(5)}_{\Lol }}
\newcommand{\dssetalga}{Y^{(p)}}
\newcommand{\dssetalgb}{Y^{(p)}_1}

\newcommand{\dssetlog}{X_1^{(5)}}

\newcommand{\nsetab}{\mathcal{P}([0, \infty))} % Normal SET for case i = 1,2 and j = 0
\newcommand{\nsetcde}{X}
\newcommand{\nsetabb}{\mathcal{X}_1} % Normal SET for case i = 1,2 and j = 2 (and 3)
\newcommand{\nsetcdeb}{X_1}
\newcommand{\nsetcdeblol}{X_{\Lol }}

\let\weakto\rightharpoonup

\title[Upscaling of dislocation walls]{Upscaling of dislocation walls in finite domains}
\author[P. van Meurs et al.]{%
  P.\ns V\ls A\ls N\ns M\ls E\ls U\ls R\ls S$\,^{1,3}$,\ns
  A.\ns M\ls U\ls N\ls T\ls E\ls A\ls N$\,^{1,2}$\ns
\and
  M.\ns A.\ns P\ls E\ls L\ls E\ls T\ls I\ls E\ls R$\,^{1,2}$
}
\affiliation{%
  $^1\,$Centre for Analysis, Scientific computing and Applications (CASA), \\
  Department of Mathematics and Computer Science, \\
  Eindhoven University of Technology, \\
  P.O. Box 513, \\
5600 MB Eindhoven, The Netherlands \\
  $^2\,$Institute for Complex Molecular Systems (ICMS), \\
  Eindhoven University of Technology, \\
$^3\,$email\textup{\nocorr: \texttt{p.j.p.v.meurs@tue.nl}} (corresponding author)\\
}
\date{\today}
\begin{document}

\maketitle

\section*{Abstract}

\begin{abstract}
We wish to understand the macroscopic plastic behaviour of metals by upscaling the micro-mechanics of dislocations. We consider a highly simplified dislocation network, which allows our microscopic model to be a one dimensional particle system, in which the interactions between the particles (dislocation walls) are singular and non-local.

As a first step towards treating realistic geometries, we focus on finite-size effects rather than considering an infinite domain as typically discussed in the literature. We derive effective equations for the dislocation density by means of $\Gamma$-convergence on the space of probability measures. Our analysis yields a classification of macroscopic models, in which the size of the domain plays a key role.
\end{abstract}

\begin{keywords} % in the same style as the 2 latest articles in EJAM
Plasticity; Multiscale; Straight edge-dislocations; Discrete-to-continuum limit; $\Gamma$-convergence; 74Q05, 74C05, 82B21, 49J45, 82D35
\end{keywords}

\section{Introduction}

% importance and phys BG
Dislocations in metals are curve-like defects in the atomic lattice of the metal. Typical metals have \emph{many} dislocations (as much as $1000$ km of dislocation curve in a cubic millimeter \cite[p.~20]{b6}), and their collective motion is the microscopic mechanism behind macroscopic \emph{permanent} or \emph{plastic} deformation.

At scales of millimeters or more, plastic deformation is well described by continuum-level theories (see e.g.~\cite[Ch.6]{Callister07}); at scales of $1{-}100$ $\mu$m, however, the specimen size, material grain size, and dislocation distribution scales become comparable, and these high-level theories break down.
At these smaller scales,  \emph{crystal plasticity} models attempt to capture the interaction between dislocations and grain boundaries by including additional degrees of freedom representing dislocation densities.

Although more detailed, such (meso-scale) crystal-plasticity models depend on closing the BBGKY hierarchy of multi-point correlation functions at the two-point or three-point correlation levels. Current methods commonly postulate a \emph{closure assumption} involving certain averages, and estimate these averages from the statistics of smaller-scale models. see for instance~\cite{ZaiserMiguelGroma01,GromaCsikorZaiser03,YefimovGromaVanderGiessen04,DengEl-Azab07,LimkumnerdVan-der-Giessen08,DengEl-Azab09}.

While this statistical approach makes sense from a practical point of view, the question remains whether microscopic models of dislocations could not be scaled up rigorously, without \emph{ad hoc} closure assumptions---although possibly in a simpler setup. Such a rigorous upscaling has been performed, for instance, for the case of parallel dislocations  \emph{on a single slip plane}, represented by a  `queue' of points on the real line~\cite{GarroniMuller05,FocardiGarroni07,b2,Monneau20091,Monneau20092}, for arbitrary  planar dislocations~\cite{KoslowskiOrtiz04,CacaceGarroni09}, and for arbitrary collections of parallel dislocations~\cite{CermelliLeoni06,Ponsiglione07,GarroniLeoniPonsiglione08TR}.

However, these upscaling techniques fail to capture one of the more intriguing aspects of interaction dislocations: the cancellation that takes place in \emph{pile-ups} of edge dislocations at grain boundaries. Roy \emph{et al.} pointed out~\cite{b1} that the stresses in such pile-ups are very sensitive to the local stacking of the dislocations, leading to incorrect predictions if the averaging is not done correctly. This may also be the reason why there are multiple, mutually contradicting dislocation-density models in the literature (e.g.~\cite{b32b22,GromaCsikorZaiser03,EversBrekelmansGeers04}).

Sparked by this observation, Scardia \emph{et al.} analyzed the structure of pile-ups in detail~\cite{b32,b34}, and showed that five different regimes exist, depending on a local aspect ratio (see also~\cite{Hall20101, Hall20102, chap2009} for an analysis of one of these regimes using formal asymptotics). We describe the results of~\cite{b32,b34} in detail below.

The authors of~\cite{b32,b34} made several simplifying assumptions, one of which is to allow the dislocations to move in a half-infinite domain. Since dislocation-density models aim to describe the cases where grain size and pile-up width are comparable, a finite domain  bounded by grain boundaries on both sides is more natural.
In this paper, we therefore generalize the results of~\cite{b32,b34} by considering any \emph{finite} length for the domain in which the dislocations are situated. This brings us to our main research question:

\medskip
%%% we are at the neck of the funnel %%%
% our main result / the subject of the paper
\begin{center}
\text{How do finite domains change the results from \cite{b32,b34}?}
\end{center}

\bigskip

%%% expanding the funnel again %%%
% what we learned
%While obtaining the answer to this question in Theorem \ref{thm intro}, we learned several things. First, by taking the limit of the length of the system size to infinity, we recover the results from \cite{b32}. Second, by making this length scale small, the finiteness of the domain becomes the dominating effect that keeps the dislocations together. Third, our results show how both these effects are combined for moderate length scales.
%
%% so what?
%Besides giving insight in the effects of the finite domain on the system, our results also provide a sound basis towards introducing a more realistic microscopic model at a later stage which includes negative edge dislocations. This will bring us again a step closer in solving the disagreement among the existing models. Once this disagreement is solved, we will have a reliable model for understanding the size-effect, which will improve the current understanding of plasticity.
%
%Furthermore, we hope that our work is inspiring for research on other particle systems as well.

After introducing our microscopic model (Section \ref{sec intro 2}), we describe the upscaling procedure in Section \ref{sec intro 4}. Then we state Theorem \ref{thm intro} - our main result - and how to interpret it from a practical point of view.

\subsection{Setting of the microscopic energy} \label{sec intro 2}

Inspired by \cite{b1}, we consider the dislocations to be arranged equidistantly in $n+1$ vertical \emph{walls} of dislocations, which are assumed to be infinitely long. Figure \ref{fig configuration model} shows a schematic picture of this configuration.

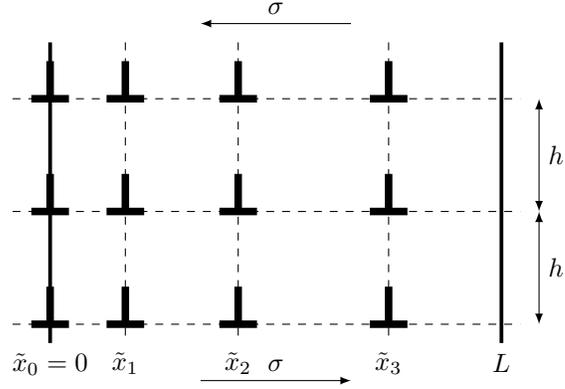
\begin{figure}[ht]
\begin{tikzpicture}[scale=0.5, >= latex]
    % barriers
    \draw[line width = 0.5mm] (0, 7.5) -- (0, -0.5) node [below] {$\tilde{x}_0 = 0$};
    \draw[line width = 0.5mm] (12, 7.5) -- (12, -0.5) node [below] {$L$};

    % \tilde{x}_i values
    \draw[dashed] (2, 7.5) -- (2, -0.5) node [below] {$\tilde{x}_1$};
    \draw[dashed] (5, 7.5) -- (5, -0.5) node [below] {$\tilde{x}_2$};
    \draw[dashed] (9, 7.5) -- (9, -0.5) node [below] {$\tilde{x}_3$};

    % dlc's, walls and slip planes
    % \x = 0
    \foreach \y in {0,1,2}
        {
        \draw[dashed] (-1, 3*\y) -- (12.5, 3*\y);

        \draw[line width = 1mm] (0, 3*\y) -- (0, 3*\y + 1);
        \draw[line width = 1mm] (0 - 0.5, 3*\y) -- (0 + 0.5, 3*\y);
        }
    \foreach \x in {2,5,9}
      \foreach \y in {0,1,2}
        {
        \draw[line width = 1mm] (\x, 3*\y) -- (\x, 3*\y + 1);
        \draw[line width = 1mm] (\x - 0.5, 3*\y) -- (\x + 0.5, 3*\y);
        }

    % arrows
    \draw[<-] (4,8)--(8,8) node[midway, above]{$\sigma$};
    \draw[->] (4,-1.5)--(8,-1.5) node[midway, above]{$\sigma$};
    \draw[<->] (13,0)--(13,3) node[midway, right]{$h$};
    \draw[<->] (13,3)--(13,6) node[midway, right]{$h$};
\end{tikzpicture}
\caption{Configuration of dislocations in the microscopic model.}
\label{fig configuration model}
\end{figure}

In the steady state, we obtain the positions of the dislocation walls, denoted by $\tilde{x}_1 \leq \tilde{x}_2 \leq \ldots \leq \tilde{x}_n$, by minimizing the energy given by
\begin{gather} \label{for ene dimful intro}
    \mathcal{E} = \Enitaintro + \EnFintro + \EnLintro,\\\notag
    \begin{aligned}
        \Enitaintro (\textbf{x}) &= K \sum_{k=1}^n \sum_{j = 0}^{n - k} V \bigghaa{ \pi \frac{ \tilde{x}_{j+k} - \tilde{x}_j }{h} }, \\
        \EnFintro (\textbf{x}) &= \sigma \sum_{i = 1}^{N} \tilde{x}_i, \\
        \EnLintro (\textbf{x}) &= \left\{
                                    \begin{array}{ll}
                                      0, & \hbox{if $\tilde{x}_n \leq L$,} \\
                                      \infty, & \hbox{otherwise.}
                                    \end{array}
                                  \right.
    \end{aligned}
\end{gather}
Here, $V$ is the \emph{interaction potential} between walls, which is defined by
\begin{equation} \label{for defn V}
V(r) := r \coth r - \log \lrabs{ \sinh r } - \log 2, r\in \mathbb{R}.
\end{equation}
The potential~$V$ is even, has a logarithmic singularity at the origin, and is strictly convex and monotonic on $(-\infty,0)$ and $(0,\infty)$. The energy~$\mathcal{E}$ involves five model parameters: $n$, the number of walls minus $1$; $h$, the distance between two subsequent dislocations in a wall; $\sigma$, a constant external load applied to the system; $L$, the position of the right boundary; $K$, a material constant.

Let us explain our model in terms of the expression for $\mathcal{E}(\textbf{x})$. The interaction part~$\Enitaintro$ is minimized by spreading the walls far apart in the interval $(0, \infty)$. The $0$ is due to the pinned wall at the impenetrable barrier located at $\tilde{x}_0 = 0$. Due to the logarithmic singularity of $V$ at $0$, none of the other walls will be located at $\tilde{x}_0$. The parts coming from the external load $\EnFintro$ and from the right impenetrable barrier $\EnLintro$ are minimized by putting the walls close to $0$. The unique minimizer (see Proposition \ref{prop strict cvy of En}) of $\mathcal{E}$ balances these effects. A thorough understanding of this balance will explain how the finiteness of the domain changes the results from \cite{b32}.

\subsection{Upscaling} \label{sec intro 4}

As mentioned in the introduction, the collective behaviour of dislocation walls will be obtained by scaling up the system described above, resulting in an energy functional $E$ which depends only on a dislocation density $\mu$. For this we need to define what it means for $\mu$ to be ``close to" a vector $\tilde{\mathbf{x}}$ of discrete wall positions. We do this by using the \emph{narrow topology}. Setting
\begin{equation*}
    \mu_n = \frac1n \sum_{j=1}^n \delta_{\tilde{x}_j}.
\end{equation*}
we say that $\mu_n$ converges in the narrow topology to $\mu$ if and only if
\begin{equation} \label{for defn intro narrow convergence}
    \intabx{[0,\infty)}{}{ \varphi }{ \mu_n } \rightarrow \intabx{[0,\infty)}{}{ \varphi }{\mu} \hsf \textrm{for all } \varphi \in C_b ([0,\infty)).
\end{equation}

As $V$ has a logarithmic singularity at $0$, the energy landscape of $\mathcal{E}$ contains $\mathcal{O}(n)$ singularities. Hence $\mathcal{E}$ will never be close to any limiting energy $E$ in any $L^\infty$-topology. Instead, we aim to prove that $\mathcal{E}$ \emph{$\Gamma$-converges} to $E$ provided that an appropriate scaling is applied. With $\Gamma$-convergence, we can show that the minimizer of $\mathcal{E}$ is close to a minimizer of $E$. Furthermore, $\Gamma$-convergence is robust to a perturbation by a continuous functional (which may model another type of external force term, for example).

From now on, all the parameters ($L$, $h$, $K$, $\sigma$) depend on $n$. In order to obtain a meaningful limit we rescale the positions $\mathbf x$ and the energy $\mathcal E$ in some $n$-dependent manner. There are two natural length scales for the rescaling of $\mathbf x$, one given by the size $L_n$ of the domain, and the other provided by an intrinsic scale arising from the balance between the load parameter $\sigma_n$ and the interaction term $\Enitaintro$.

Inspired by~\cite{b32} we define this second length scale as
\begin{equation*} \notag
    \ell_n = \frac{n h_n}\pi \ahat_n,
\end{equation*}
where $\ahat_n$ is a parameter which scales like the \emph{aspect ratio} between the dislocations in Figure \ref{fig configuration model}, i.e. the typical horizontal distance between walls divided by $h_n$. It depends on the parameters in the following way:
\begin{gather} \label{for defn intro ahat}
\ahat_n := f_n \lrhaa{ \sqrt{ \frac{ \pi K_n }{ n \sigma_n h_n } } } \textrm{, with} \\\label{for defn intro f n}
f_n(a) := \left\{
           \begin{array}{ll}
             n a^2, & \hbox{$a < \dfrac1n$,} \\
             a, & \hbox{$\dfrac1n \leq a \leq 1$,} \\
             \log a + 1, & \hbox{$1 < a$.}
           \end{array}
         \right.
\end{gather}
Figure \ref{fig plot f} illustrates the typical behaviour of $f_n$. We define the ratio
\begin{equation} \label{for defn Lol n}
    \Lol{}_n := \frac{L_n}{\ell_n}
\end{equation}
to characterize the relative size of $\ell_n$ and $L_n$.

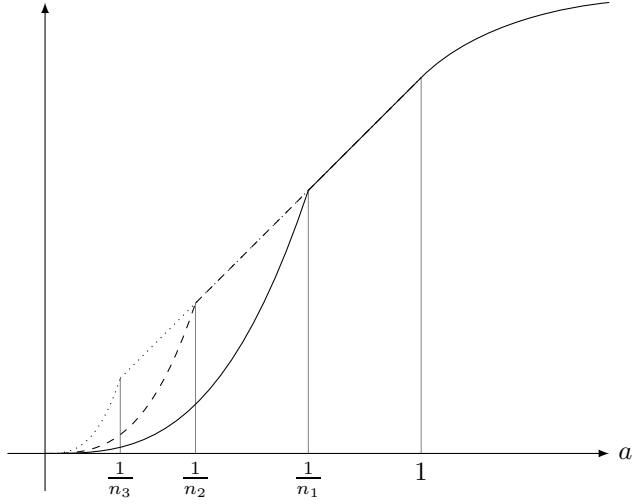
\begin{figure}
\begin{tikzpicture}[scale=5.0, >= latex] % see a.3.72
\draw[very thin, color=gray] (1,0) node[below,color=black] {$1$} -- (1,1);
\draw[very thin, color=gray] (0.7,0) node[below,color=black] {$\frac 1{n_1}$} -- (0.7,0.7);
\draw[very thin, color=gray] (0.4,0) node[below,color=black] {$\frac 1{n_2}$}-- (0.4,0.4);
\draw[very thin, color=gray] (0.2,0) node[below,color=black] {$\frac 1{n_3}$}-- (0.2,0.2);
\draw[->] (-0.1,0) -- (1.5,0) node[right] {$a$};
\draw[->] (0,-0.1) -- (0,1.2);
\draw[dotted, domain = 0:0.2] plot ( \x, { (\x)^3 / 0.04 } );
\draw[dotted, domain = 0.2:1] plot ( \x, \x );
\draw[dashed,domain=0:0.4]    plot ( \x, { (\x)^3 / 0.16 } );
\draw[dashed,domain=0.4:1]    plot ( \x, \x );
\draw[domain=0:0.7]           plot ( \x, { (\x)^3 / 0.49 } ); % node[right] {$f_{n_1}(a)$}
\draw[domain=0.7:1]           plot ( \x, \x );
\draw (1,1) .. controls (1.15,1.15) and (1.4,1.19) .. (1.5,1.2);
\end{tikzpicture}
 \caption{\footnotesize Plots of $f_n$ (see \eqref{for defn intro f n}) for $n_1 < n_2 < n_3$.}
 \label{fig plot f}
\end{figure}

Whenever $\ell_n$ is asymptotically smaller than $L_n$, i.e.\ $\Lol{}_n\gg 1$, it is natural to rescale the positions by $\ell_n$. In this case the scaled energy is given by
\begin{equation} \label{for ene dimless intro}
E_n (\mathbf{x}^n) := \left\{
                        \begin{array}{ll}
                          \Bstrut \displaystyle \frac1{ n \sigma_n \ell_n } \mathcal{E} \lrhaa{ \ell_n {\mathbf{x}}^n } - \frac12 \log \frac{e}{ 2 n \ahat_n }, & \hbox{if $\ahat_n \ll 1/n$,} \\
                          \displaystyle \frac1{ n \sigma_n \ell_n } \mathcal{E} \lrhaa{ \ell_n {\mathbf{x}}^n }, & \hbox{otherwise.}
                        \end{array}
                      \right.
\end{equation}
The $\Gamma$-convergence result of $E_n$ to $E$ is stated in [\nbcite{b32}, Theorem 1]. There are five expressions for the related limiting energy $E$, depending on which of five scaling regimes $\ahat_n$ belongs to. We come back to this while discussing Table \ref{tab Eita}.

On the other hand, when $L_n \lesssim \ell_n$, i.e. $\Lol{}_n\lesssim 1$, the barrier at $L_n$ is likely to determine the typical length scale for $\tilde{\mathbf{x}}$, and we scale $\tilde{\mathbf{x}}$ with $L_n$. The expression for the aspect ratio then also changes:
\begin{equation} \label{for defn intro alpha}
    \alpha_n := \frac{\pi L_n}{n h_n} = \Lol_n \ahat_n.
\end{equation}
In this case (i.e. $\Lol_n \ll 1$ or $\Lol_n \sim 1$), we scale the energy as follows:
\begin{equation} \label{for ene dimless L intro}
E_n (\mathbf{x}^n) := \left\{
                        \begin{array}{ll}
                          \Bstrut \displaystyle \frac{\Lol_n}{ n \sigma_n L_n } \mathcal{E} \lrhaa{ L_n {\mathbf{x}}^n } - \frac12 \log \frac{e}{ 2 n \alpha_n }, & \hbox{if $\alpha_n \ll 1/n$,} \\
                          \Bstrut \displaystyle \frac{ \exp \bighaa{ 2 \alpha_n (1 - 1/\Lol_n)} }{ n \sigma_n L_n } \mathcal{E} \lrhaa{ L_n {\mathbf{x}}^n }, & \hbox{if $\alpha_n \gg 1$,} \\
                          \displaystyle \frac{\Lol_n^2}{ n \sigma_n L_n } \mathcal{E} \lrhaa{ L_n {\mathbf{x}}^n }, & \hbox{otherwise.}
                        \end{array}
                      \right.
\end{equation}

In order to state the main result we extend $E_n$ to apply to measures by setting
\begin{equation} \label{for ene dimless extension intro}
    E_n (\mu) = \left\{
                  \begin{array}{ll}
                    E_n (\mathbf{x}^n), & \hbox{if \ddollar{ \mu = \frac1n \sum_{j=1}^n \delta_{x_j} },} \\
                    \infty, & \hbox{otherwise.}
                  \end{array}
                \right.
\end{equation}

\begin{thm} \label{thm intro}

(Convergence of the energy).
Let $\alpha_n$ and $\Lol_n$ be such that they satisfy any of the criteria as in the first columns of Table \ref{tab Eita} and Table \ref{tab EF and EL}. Then boundedness of $E_n (\mu_n)$ (as in \eqref{for ene dimless extension intro}) implies that $(\mu_n)$ is compact in the narrow topology. Moreover, $E_n$  $\Gamma$-converges with respect to the narrow topology to
\begin{equation*}
    E = \Eitaintro + \EFintro + \ELintro,
\end{equation*}
where the components are given in Table \ref{tab Eita} and Table \ref{tab EF and EL}, except for the particular case in which $1 \ll \alpha_n$ and $\exp \bighaa{ 2 \alpha_n (1 - 1/\Lol_n) } \rightarrow \infty$, which is treated in Theorem \ref{thm particular case}.

\end{thm}

\begin{table}
\caption{\footnotesize Expressions for $\Eitaintro$, the interaction part of the limit energy. If $\Lol_n \gg 1$, one has to read $\ahat_n$ instead of $\alpha_n$.}
\begin{tabular}{clc}
  \hline \hline
  regime & $\Eitaintro (\mu)$ & $p$ \\\hline
  \Bstrut $\alpha_n \ll \dfrac1n$ & \ddollar{ \frac{1}{2} \intabxnohs{0}{\infty}{ \intabx{0}{\infty}{ \log \frac1{|x - y|} }{\mu (y)} }{ \mu(x) } } & 1 \\
  \Bstrut $n \alpha_n \rightarrow \tilde{c}$ & \ddollar{ \frac{\tilde{c}}{2} \intabxnohs{0}{\infty}{ \intabx{0}{\infty}{ V ( \tilde{c}(x - y) ) }{\mu (y)} }{ \mu(x) } } & 2 \\
  \Bstrut $\dfrac1n \ll \alpha_n \ll 1$ & \ddollar{ \left\{
                     \begin{array}{ll}
                       \displaystyle \lrhaa{ \int_0^\infty V } \int_0^\infty \rho^2, & \hbox{if $d\mu(x) = \rho(x) dx$,} \\
                       \infty, & \hbox{otherwise}
                     \end{array}
                   \right. } & 3 \\
  \Bstrut $\alpha_n \rightarrow \tilde{c}$ & \ddollar{ \left\{
                     \begin{array}{ll}
                       \displaystyle \tilde{c} \intabx{0}{\infty}{ \bigghaa{ \sum_{k = 1}^\infty V \Bighaa{ k \frac{\tilde{c}}{ \rho(x) } } } \rho(x) }{x}, & \hbox{if $d\mu(x) = \rho(x) dx$,} \\
                       \infty, & \hbox{otherwise}
                     \end{array}
                   \right. } & 4 \\
  $1 \ll \alpha_n$ & \ddollar{ \left\{
                     \begin{array}{ll}
                       2 e^{-2} \indicator{ \Lol \leq 1 } , & \hbox{if \ddollar{ \ga{\mu}{\mathcal{L}} = \rho \leq 1} \hsone $\mathcal{L}$-a.e.,} \\
                       \infty, & \hbox{otherwise}
                     \end{array}
                   \right. } & 5 \\
  \hline \hline
\end{tabular}
\label{tab Eita}
\end{table}

\begin{table}
\caption{\footnotesize Expressions for $\EFintro$ and $\ELintro$, the parts in the limit energy coming from the external force and the second barrier. The constant $C$ is given by \eqref{for defn intro C}.}
\begin{tabular}{cllc}
  \hline \hline
  regime & $\EFintro (\mu)$ & $\ELintro (\mu)$ & $q$\\\hline
  \Bstrut $\Lol{}_n \gg 1$ & \ddollar{ \intabx{0}{\infty}{ x }{ \mu(x) } } & \ddollar{ 0 } & 1 \\
  \Bstrut $\Lol{}_n \rightarrow \Lol{}$ & \ddollar{ C \bighaa{ \Lol ; (\alpha_n) } \intabx{0}{\infty}{ x }{ \mu(x) } } & \ddollar{ \left\{
                     \begin{array}{ll}
                       0, & \hbox{if $\supp \mu \subset [0, 1]$,} \\
                       \infty, & \hbox{otherwise.}
                     \end{array}
                   \right. }
                   &2\\
  $\Lol{}_n \ll 1$ & \ddollar{ 0 } & \ddollar{ \left\{
                     \begin{array}{ll}
                       0, & \hbox{if $\supp \mu \subset [0, 1]$,} \\
                       \infty, & \hbox{otherwise.}
                     \end{array}
                   \right. } &3\\
  \hline \hline
\end{tabular}
\label{tab EF and EL}
\end{table}

The state of the art before this paper is given by Table~\ref{tab Eita}; Table~\ref{tab EF and EL} shows our generalization of the results of~\cite{b32} to finite domains. For a given set of parameters ($n$, $L_n$, $h_n$, $K_n$, $\sigma_n$), we can calculate $\ell_n$ and consecutively $\Lol_n$ and $\alpha_n$, and thus we know \textit{a priori} which of the expressions for $\Eitaintro (\mu)$, $\EFintro (\mu)$ and $\ELintro (\mu)$ we have as limit energy.
%If it so happens that either $\Lol_n$ or $\alpha_n$ does not satisfy any of the regimes as listed in the first columns of Table \ref{tab Eita} or Table \ref{tab EF and EL}, Theorem \ref{thm intro} is not applicable. However, in this case one cannot make sense of a unique $\Gamma$-limit (see Remark \ref{rem motivation conditions main thm}).

\medskip
In all cases the limit energy gives rise to a well-posed variational problem: minimizers exist and are unique (Theorem~\ref{thm main appl}). By the usual results on $\Gamma$-convergence, minimizers are the limit of the sequence of the finite-$n$ minimizers (Corollary \ref{cor main thm}).

\subsection{Discussion}
\label{sec:discussion-intro}
We started with the question how the finiteness of the domain changes the results from \cite{b32}. We now discuss the assertions of Theorem~\ref{thm intro} from this viewpoint, for which we use a schematic plot of the parameter space (Figure \ref{fig pm plot}).

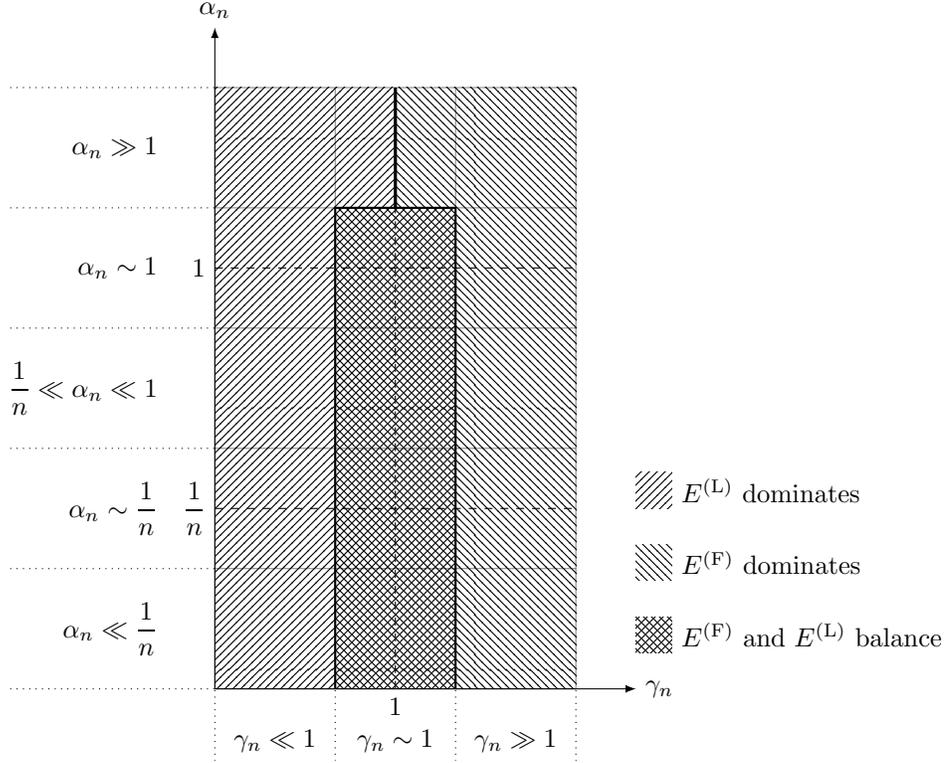
\begin{figure}
\begin{tikzpicture}[scale=1.6, >= latex]
% axis
\draw[very thin,color=gray] (0,0) grid (3,5);
\draw[->] (0,0) -- (3.5,0) node[right] {$\Lol_n$};
\draw[->] (0,0) -- (0,5.5) node[above] {$\alpha_n$};
% specific parameter values
\draw[dashed] (0, 1.5) node[left] {$\dfrac 1n$} -- (3, 1.5);
\draw[dashed] (0, 3.5) node[left] {$1$} -- (3, 3.5);
\draw[dashed] (1.5, 0) node[below] {$1$} -- (1.5, 4);
% pm declaration
\pgfmathsetmacro{\tikzPmplotLeftMarge}{1.7}
\pgfmathsetmacro{\tikzPmplotBelowMarge}{0.6}
\pgfmathsetmacro{\tikzPmplotLeftText}{-0.4}
\pgfmathsetmacro{\tikzPmplotBelowText}{-0.6}
\pgfmathsetmacro{\tikzPmplotLineHeight}{0.3}
% extra lines
\foreach \y in {0,1,2,3,4,5}
    \draw[dotted] (-\tikzPmplotLeftMarge, \y) -- (0, \y);
\foreach \x in {0,1,2,3}
    \draw[dotted] (\x, -\tikzPmplotBelowMarge) -- (\x, 0);
% pm regime labels
\draw[white] (0, \tikzPmplotBelowText) -- (1, \tikzPmplotBelowText) node[midway, above, black]{$\Lol_n \ll 1$};
\draw[white] (1, \tikzPmplotBelowText) -- (2, \tikzPmplotBelowText) node[midway, above, black]{$\Lol_n \sim 1$};
\draw[white] (2, \tikzPmplotBelowText) -- (3, \tikzPmplotBelowText) node[midway, above, black]{$\Lol_n \gg 1$};
\draw[white] (\tikzPmplotLeftText, 0) -- (\tikzPmplotLeftText, 1) node[midway, left, black]{$\alpha_n \ll \dfrac 1n$};
\draw[white] (\tikzPmplotLeftText, 1) -- (\tikzPmplotLeftText, 2) node[midway, left, black]{$\alpha_n \sim \dfrac 1n$};
\draw[white] (\tikzPmplotLeftText, 2) -- (\tikzPmplotLeftText, 3) node[midway, left, black]{$\dfrac 1n \ll \alpha_n \ll 1$};
\draw[white] (\tikzPmplotLeftText, 3) -- (\tikzPmplotLeftText, 4) node[midway, left, black]{$\alpha_n \sim 1$};
\draw[white] (\tikzPmplotLeftText, 4) -- (\tikzPmplotLeftText, 5) node[midway, left, black]{$\alpha_n \gg 1$};
% fat lines
\draw[thick] (1, 0) -- (1, 4) -- (2, 4) -- (2, 0);
\draw[very thick] (1.5, 4) -- (1.5, 5);
% hatching
\fill[pattern = north east lines] (0, 0) rectangle (1.5,5);
\fill[pattern = north east lines] (1.5, 0) rectangle (2,4);
\fill[pattern = north west lines] (1, 0) rectangle (1.5,4);
\fill[pattern = north west lines] (1.5, 0) rectangle (3,5);
% legend
\draw[white] (3.5 + \tikzPmplotLineHeight, \tikzPmplotLineHeight) -- (3.5 + \tikzPmplotLineHeight, 2*\tikzPmplotLineHeight) node[midway, right, black]{$\EFintro$ and $\ELintro$ balance};
\draw[white] (3.5 + \tikzPmplotLineHeight, 3*\tikzPmplotLineHeight) -- (3.5 + \tikzPmplotLineHeight, 4*\tikzPmplotLineHeight) node[midway, right, black]{$\EFintro$ dominates};
\draw[white] (3.5 + \tikzPmplotLineHeight, 5*\tikzPmplotLineHeight) -- (3.5 + \tikzPmplotLineHeight, 6*\tikzPmplotLineHeight) node[midway, right, black]{$\ELintro$ dominates};
\fill[pattern = north east lines] (3.5, \tikzPmplotLineHeight) rectangle (3.5 + \tikzPmplotLineHeight, 2*\tikzPmplotLineHeight);
\fill[pattern = north west lines] (3.5, \tikzPmplotLineHeight) rectangle (3.5 + \tikzPmplotLineHeight, 2*\tikzPmplotLineHeight);
\fill[pattern = north west lines] (3.5, 3*\tikzPmplotLineHeight) rectangle (3.5 + \tikzPmplotLineHeight, 4*\tikzPmplotLineHeight);
\fill[pattern = north east lines] (3.5, 5*\tikzPmplotLineHeight) rectangle (3.5 + \tikzPmplotLineHeight, 6*\tikzPmplotLineHeight);
\end{tikzpicture}
 \caption{\footnotesize Plot of the regions in parameter space in which either $\ELintro \ll  \EFintro$, $\EFintro \ll \ELintro$, or~$\EFintro \sim \ELintro$. The axes show the asymptotic behaviour of $\Lol_n$ and $\alpha_n$. Although the parameter space is divided in a matrix of five by three blocks, their boundaries do not correspond to specific scalings of $\alpha_n$ or $\Lol_n$. }
 \label{fig pm plot}
\end{figure}

\begin{itemize}
  \item First note that if $\Lol_n \gg 1$, i.e. $L_n\gg \ell_n$, then we recover the same limit energy as in \cite{b32}. This can be considered a consistency check, showing that the results of this paper generalize~\cite{b32}.
  \item Moving away from the case of~\cite{b32}, the case $\Lol_n \ll 1$ is the simplest: here the finiteness of the domain completely dominates the external forcing (first column in Figure \ref{fig pm plot}). The scaling is independent of the external forcing, and the limit energy is governed by the balance between the interactions and the finiteness of the domain.
  \item The critical case $\Lol{}_n \rightarrow \Lol{}$ is more subtle (second column in Figure \ref{fig pm plot}), as can be recognized e.g. in the constant that multiplies the force term of the limit energy. This constant is given by
\begin{gather}\label{for defn intro C}
    C \bighaa{ \Lol ; (\alpha_n) } := \left\{
                            \begin{array}{ll}
                              \Lol, & \hbox{if $\alpha_n \ll 1/n$,} \\
                              \limLolone{} / 2, & \hbox{if $\alpha_n \gg 1$,} \\
                              \Lol^2, & \hbox{otherwise,}
                            \end{array}
                          \right. \\
\text{where} \hsf \limLolone{} := \lim_{n \rightarrow \infty} \exp \bighaa{ 2 \alpha_n (1 - 1/\Lol_n) }. \label{for defn intro limLolone}
\end{gather}
It describes the transition between $\ELintro\ll \EFintro$ (i.e. $C \bighaa{ \Lol ; (\alpha_n) } = \infty$), to $\EFintro\ll\ELintro$ (i.e. $C \bighaa{ \Lol ; (\alpha_n) } = 0$). When $C \bighaa{ \Lol ; (\alpha_n) } \in (0, \infty)$, both terms of the energy contribute a finite amount. Indeed, for these values of $C \bighaa{ \Lol ; (\alpha_n) }$ we could have chosen the scaling for $\tilde{\mathbf{x}}$ to be as in \eqref{for ene dimless intro} as well. The $\Gamma$-limit would contain just as much information. However, we use the other scaling \eqref{for ene dimless L intro} for purely practical reasons.

When $\alpha_n \gg 1$, the transition is very delicate: $C \bighaa{ \Lol ; (\alpha_n) } \in (0, \infty)$ if $\limLolone \in (0, \infty)$, which can only occur if $\Lol = 1$. This is indicated in Figure \ref{fig pm plot} by the vertical line at~${ \Lol_n \rightarrow 1 }$. If $\limLolone = \infty$, it holds that $\EFintro (\mu) = \infty$, and hence the scaling of $\tilde{\mathbf{x}}$ by $L_n$ doesn't give a useful limit energy. That is why the case $\limLolone = \infty$ is excluded in Theorem~\ref{thm intro}. The scaling as given by \eqref{for ene dimless intro} does work. This is made precise by Theorem~\ref{thm particular case}, from which we conclude that $\ELintro$ is indeed negligible with respect to $\EFintro$ in this case.
%$\Lol = 1$ is a specific case; only $\mu = L |_{[0,1]}$ will give a finite value for the limit energy. One can see this from Table \ref{tab Eita} and Table \ref{tab EF and EL} by the constant $2 e^{-2}$ in the expression for $\Eitaintro $, and by the constant $C \bighaa{ \Lol ; (\alpha_n) }$ in $\EFintro $.

  \item For $\alpha_n \gg 1$ and $\Lol_n$ bounded such that $\limLolone{} \neq \infty$ (i.e. the part of parameter space given by the left half of the first row of Figure \ref{fig pm plot}), the energy $E$ is degenerate in the sense that it is only finite at exactly one point, the measure $\mu = \mathcal{L} |_{(0, 1)}$. Hence it only contains information about the minimizer. One way to obtain more information in the limit energy is by using a logarithmic rescaling of $E_n$. In Theorem~\ref{thm log} we state our result that
\begin{equation*}
    \frac 1{2 \alpha_n} \log E_n \Gto \lrhaa{ \mu \mapsto 1 - \biggnorm{ \ga{\mu}{\mathcal{L}}}{\infty}^{-1} }.
\end{equation*}

  \item It might be instructive to note that the five expressions for $\Eitaintro (\mu)$ only depend on $\Lol_n$ through the choice of rescaling with $\ahat_n$ versus $\alpha_n$. This shows that the presence of the second barrier does not influence the interaction behaviour of the walls.

\end{itemize}

Summarizing, the finiteness of the domain induces a second length scale---the length of the domain $L_n$---in addition to the length scale $\ell_n$ generated by the external forcing. We specified three qualitatively different limiting behaviours for the energy, which correspond to the cases $L_n$ being asymptotically bigger, smaller, or equal to $\ell_n$. This result enables us to test the mutually contradicting dislocation-density models (as mentioned in the introduction) with more freedom in the microscopic setting. As a special case, we are able to test these models when no loading is applied (i.e. $\sigma_n = 0$).

On the other hand, for the parameter regime in which the forcing term is negligible with respect to the effect of the finite domain, it seems unphysical to ignore the effect of \emph{negative} edge dislocations. One of the reasons that we do not consider a model with negative edge dislocations, is that the effect of nucleation and annihilation of dislocations with opposite sign results in an energy that is not bounded from below. Various methods have been used to circumvent this issue~\cite{MichaelThesis,CaiArsenlisWeinbergerBulatov06}, but they each have their drawbacks. We plan to explore the extension of the present results to the case of multiple signs in the future.

A significant step towards applicability would be to replace the assumption of equi-spaced slip planes by a stochastic spacing, as also suggested by other authors~\cite{b32b25,b27}). If one maintains the wall assumption, then stochastic spacing leads to a different interaction potential $V$, for which no explicit expressions exist, and for which the large-distance behaviour is not yet completely understood~\cite{b32b25}. However, in the case of stochastically spaced slip planes, dislocations do not form exact walls~\cite{Portegies12TR} so that one requires a fully two-dimensional description. A rigorous upscaling in the two-dimensional framework would be the ultimate goal, but that is still far away.

Besides extending the microscopic model to have more freedom in space, one can also consider dislocation \emph{dynamics}, which is paramount for understanding plasticity. In the case of a linear drag law~\cite[Ch.~7]{HirthLothe82}, these dynamics are described by a \emph{gradient flow} of the energy. Upscaling the dynamics of the discrete dislocation walls to dislocation densities requires more than just $\Gamma$-convergence of the energies (see e.g. \cite{Serfaty11}); one also needs lower bounds on the slopes. We plan to return to this question in a future publication.

This paper is organized as follows. We prove Theorem \ref{thm intro} in Section \ref{sec thm pf}, which requires a detailed description of our setting and its notation (Section \ref{sec prelim not}) followed by crucial arguments that support the proof (Section \ref{sec diag arg} up to and including Section \ref{sec ene F L}). This leaves us with the small range of parameters which is excluded in Theorem \ref{thm intro}, with the question whether the limiting energies still have a unique minimizer (and whether the discrete minimizers converge to it), and with the issue that the limiting energy in the dilute case (i.e. $\alpha_n \gg 1$ and $\Lol \leq 1$) solely contains information about the minimizer. These three issues are all separately solved in Section \ref{sec further results}. In the Appendix we discuss a few technical steps in the proof of Theorem \ref{thm intro}, and we briefly recall the definition of $\Gamma$-convergence together with its basic properties.

\section{Preliminaries}

\subsection{Notation} \label{sec prelim not}

\subsubsection{Basic notation}

\begin{itemize}
  \item We denote a sequence by $(a_n)$.
  \item \ddollar{ \chi_{ \acc{A} } := \left\{
                                        \begin{array}{ll}
                                          0, & \hbox{$A$ is true,} \\
                                          \infty, & \hbox{$A$ is false.}
                                        \end{array}
                                      \right.
   }
  \item $\overline{\R{}} := \R{} \cup \{\pm \infty\}$.
  \item We denote the Lebesgue measure by $\mathcal{L}$.
  \item For $\xi \in BV(\R{})$, we denote the distributional derivative by $D\xi$. If $\xi$ is at least weakly differentiable, we use the common notation $\xi'$ for the real-valued derivative of $\xi$.
  \item $\mathcal{P} ([0, \infty)) := $ space of probability measures.
  \item Let $X$ be a metric space and $E : X \rightarrow \overline{\R{}}$. A subset  $Y \subset X$ is said to be \textbf{energy dense} if
      \begin{equation*}
        \bigaccv{ (y, E(y)) }{ y \in Y } \subset Y \times \overline{\R{}} \textrm{ is dense in } \bigaccv{ (x, E(x)) }{ x \in X },
      \end{equation*}
      or equivalently,
      \begin{equation*}
        \forall x \in X \hsq \exists (y_n) \subset Y : y_n \rightarrow x \textrm{ and } \lim_{n \rightarrow \infty} E(y_n) = E(x).
      \end{equation*}
      The set $Y$ is said to be \textbf{lower energy dense} in $X$ with respect to $E$ if
      \begin{equation*}
        \forall x \in X \hsq \exists (y_n) \subset Y : y_n \rightarrow x \textrm{ and } \limp E(y_n) \leq E(x).
      \end{equation*}
      Note that energy density implies lower energy density. We need to prove lower energy density of two sets a number of times, but often it is just as easy to show that they are even energy dense.
  \item We use the symbols $\sim$, $\ll$ and $\lesssim$ to precisely denote the different scaling regimes for~$\Lol{}_n$ and $\ahat_n$ or $\alpha_n$. They are defined as follows. Let $(a_n), (b_n) \subset \R{}$, then
\begin{align*}
    a_n \sim b_n &\lrard{} \frac{a_n}{b_n} \textrm{ converges to some } C \in (0, \infty), \\
    a_n \ll b_n &\lrard{} \limp \frac{a_n}{b_n} = 0, \\
    a_n \lesssim b_n &\lrard{} a_n \sim b_n \textrm{ or } a_n \ll b_n.
\end{align*}
We similarly define $\gg$ and $\gtrsim$. Two sequences $(a_n), (b_n) \subset \R{}$ do not have to satisfy any of the above criteria. However, these sequences are not important to us, as we shall argue in Remark \ref{rem motivation conditions main thm}.

In the standard asymptotics literature, $\sim$ typically means $a_n/b_n \rightarrow 1$. This is expressed here by writing $a_n = b_n + \mathcal{O}(c_n)$, where a sequence $c_n \ll b_n$ is specified.
\end{itemize}

\subsubsection{Difference in notation compared to \cite{b32}}

We use a slightly different expression for $K$ and $V$ to simplify formulas. To make the connection clear, we decorate the corresponding quantities in \cite{b32} by a sub- or superscript $\mathrm{GPPS}$, in honour of the authors. The connection is given by $K = K_{\mathrm{GPPS}}/\pi^2$ and
\begin{equation*}
    V(r) := \pi^2 V_{\mathrm{GPPS}} \Bighaa{ \frac{r}{\pi} }.
\end{equation*}

\begin{prop} \label{properties V}

(Properties of the interaction potential $V$).
$V$ as defined by \eqref{for defn V} satisfies:
\begin{enumerate}[(i)]
  \item $V$ is even;
  \item $V |_{(0, \infty)}$ is strictly convex;
  \item $V(r) = \log \frac1r + 1 - \log 2 + \mathcal{O} (r^2)$ for $r \ll 1$ \label{properties V asymp 0};
  \item \label{prop asymp V large t} $V(r) = 2 r e^{- 2 r} + \mathcal{O} (r e^{- 4 r})$ for $r \gg 1$.
\end{enumerate}

\end{prop}

\subsubsection{Scaling regimes}

We use the letter $q \in \{1,2,3\}$ to indicate any of the three scaling regimes for $\Lol{}_n$ in Table~\ref{tab EF and EL}. As a result, $q$ labels the columns in Figure \ref{fig pm plot} in decreasing order. We also use~$q = 0$, which corresponds to $\Lol{}_n = \infty$, to indicate the setting without second barrier (as in \cite{b32}). Let us immediately use $q$ to unify the notation for the aspect ratio: let $\anq{q}$ be defined by
\begin{equation}
\label{for choice alpha}
    \anq{0} := \anq{1} := \ahat_n, \hsf \anq{2} := \anq{3} := \alpha_n.
\end{equation}
Similarly, we introduce $p \in \I{5}$ to indicate any of the five scaling regimes for $\ahat_n$ in Table \ref{tab Eita}. In decreasing order, $p$ labels the rows in Figure \ref{fig pm plot}. The following list illustrates how we exploit the indices $p$ and $q$ to distinguish scaling regimes:
\begin{itemize}
  \item $(p, q)$: we consider any scaling for $\anq{q}$ and $\Lol_n$ at the same time.
  \item $(p, 3)$: we consider $\Lol{}_n \ll 1$, but no restriction on the scaling of $\anq{q}$. We also refer to this by ``case $q = 3$".
  \item $(5, 0)$: $1 \ll \ahat_n$ and $\Lol{}_n \ll 1$.
  \item $(2, q)$ for $q = 2,3$: $\anq{q} \sim 1/n$ and $\Lol{}_n \lesssim 1$.
  \item $((2-4), q)$: short-hand notation for $(p, q)$ for $p = 2,3,4$. It means that $1/n \lesssim \anq{q} \lesssim 1$ and no restrictions on the scaling of $\Lol_n{}$.
\end{itemize}
Not all possible sequences $\anq{q}$ and $\Lol_n$ can be characterized by a single value for $p$ or $q$. Fortunately, the following remark shows that these sequences can never yield a unique limit for the related energy functionals.

\begin{rem} \label{rem motivation conditions main thm}
(Explanation for conditions in Theorem \ref{thm intro}).
Let $\anq{q}$ or $\Lol{}_n$ be such that they can not be characterized by a single value for $p$ or $q$. Then there exist at least two subsequences that belong to a different class (or converge to a different constant $\tilde{c}$ or $\Lol{}$). As can be seen from the expression for $E^{(p, q)} (\mu)$, this would give different limit energies, depending along which of these subsequences we take the $\Gamma$-limit, and hence the $\Gamma$-limit does not exist for such sequences $\anq{q}$ or $\Lol{}_n$.

\end{rem}

\subsubsection{Energies for fixed $n$}

From this point on, we denote the energy as stated in \eqref{for ene dimless intro} and \eqref{for ene dimless L intro} by
\begin{equation*}
E_n^{(p, q)} = \Enita{p}{q} + \EnF{p}{q} + \EnL{p}{q}: \Omega_n \rightarrow \overline{\R{}},
\end{equation*}
where
\begin{equation*}
    \Omega_n := \bigaccv{ (x_1, \ldots, x_n) \in [0, \infty)^n }{ x_1 \leq x_2 \leq \ldots \leq x_n }.
\end{equation*}
Furthermore, we define $x_0 := 0$ to indicate the pinned dislocation wall at the left barrier.

Now we can explicitly denote all the components of the energies $E_n^{(p, q)}$ in terms of the two parameters $\anq{q}$ and $\Lol{}_n$:
\label{pageref:for En}
\begin{align*}
    \Enita{1}{q} (x^n) &= \frac1{ n^2 } \sum_{k=1}^n \sum_{j = 0}^{n - k} V \bighaa{ n \anq{q} \bighaa{ x_{j+k}^n - x_j^n } } - \frac12 \log \frac{e}{ 2 n \anq{q} }, \\
    \Enita{(2-4)}{q} (x^n) &= \frac{ \anq{q} }n \sum_{k=1}^n \sum_{j = 0}^{n - k} V \lrhaa{ n \anq{q} \bighaa{ x_{j+k}^n - x_j^n } }, \\
    \Enita{5}{q} (x^n) &= \frac{ \exp \bighaa{2 \bighaa{ \anq q - 1 } } }{ n \anq q } \sum_{k=1}^n \sum_{j = 0}^{n - k} V \bighaa{ n \anq q \bighaa{ x_{j+k}^n - x_j^n } }, \\
    \EnF{p}{(0-1)} (x^n) &= \frac1n \sum_{i = 1}^{n} x_i^n, \\
\EnF{1}{(2-3)} (x^n) &= \Lol_n \frac1n \sum_{i = 1}^{n} x_i^n, \\
    \EnF{(2-4)}{(2-3)} (x^n) &= \Lol_n^2 \frac1n \sum_{i = 1}^{n} x_i^n, \\
    \EnF{5}{(2-3)} (x^n) &= \exp \lrhaa{ 2 \alpha_n \lrhaa{1 - \frac1{\Lol_n} } } \frac1n \sum_{i = 1}^{n} x_i^n, \\
    \EnL{p}{(0-1)} (x^n) &= \chi_{ \acc{ x^n_n \leq \Lol{}_n } }, \\
    \EnL{p}{(2-3)} (x^n) &= \chi_{ \acc{ x^n_n \leq 1 } },
\end{align*}% See scaling.pdf (9) and (25), and a.2.95.

%\begin{equation}\label{for En}
%\begin{aligned}
%    \Enita{1}{q} (x^n) &= \frac1{ n^2 } \sum_{k=1}^n \sum_{j = 0}^{n - k} V \bighaa{ n \anq{q} \bighaa{ x_{j+k}^n - x_j^n } } - \frac12 \log \frac{e}{ 2 n \anq{q} }, \\
%    \Enita{(2-4)}{q} (x^n) &= \frac{ \anq{q} }n \sum_{k=1}^n \sum_{j = 0}^{n - k} V \lrhaa{ n \anq{q} \bighaa{ x_{j+k}^n - x_j^n } }, \\
%    \Enita{5}{q} (x^n) &= \frac{ \exp \bighaa{2 \bighaa{ \anq q - 1 } } }{ n \anq q } \sum_{k=1}^n \sum_{j = 0}^{n - k} V \bighaa{ n \anq q \bighaa{ x_{j+k}^n - x_j^n } }, \\
%    \EnF{p}{(0-1)} (x^n) &= \frac1n \sum_{i = 1}^{n} x_i^n, \\
%\EnF{1}{(2-3)} (x^n) &= \Lol_n \frac1n \sum_{i = 1}^{n} x_i^n, \\
%    \EnF{(2-4)}{(2-3)} (x^n) &= \Lol_n^2 \frac1n \sum_{i = 1}^{n} x_i^n, \\
%    \EnF{5}{(2-3)} (x^n) &= \exp \lrhaa{ 2 \alpha_n \lrhaa{1 - \frac1{\Lol_n} } } \frac1n \sum_{i = 1}^{n} x_i^n, \\
%    \EnL{p}{(0-1)} (x^n) &= \chi_{ \acc{ x^n_n \leq \Lol{}_n } }, \\
%    \EnL{p}{(2-3)} (x^n) &= \chi_{ \acc{ x^n_n \leq 1 } },
%\end{aligned}% See scaling.pdf (9) and (25), and a.2.95.
%\end{equation}
The constant in $\Enita{1}{q} (x^n)$ is introduced to balance a constant contribution to the energy for each wall-wall interaction, regardless of their intermediate distance. From the expressions above and $V$ being strictly convex on $(0, \infty)$, it is easy to see that the following proposition holds:

\begin{prop} \label{prop strict cvy of En}
$E_n^{(p, q)} : \Omega_n \rightarrow \R{}$ is strictly convex.

\end{prop}

\subsubsection{Limit energies}

We continue with the notation for the limit energies. Let
\begin{gather*}
    E^{(p, q)} : \mathcal{P} ([0, \infty)) \rightarrow \R{}, \\
    E^{(p, q)}(\mu) = \Eita{p}{q}(\mu) + \EF{p}{q}(\mu) + \EL{p}{q}(\mu),
\end{gather*}
of which the expressions for the components are listed in Table \ref{tab Eita} and Table \ref{tab EF and EL}. To be precise, we need to define $\Lol = \infty$ in case $q = 0,1$ and $\Lol = 0$ in case $q = 3$, in order to make sense of the expression of $\Eita{5}{q}(\mu)$. Observe that in case $p = 5$ and $q = 2,3$, the expression for $E^{(p, q)}(\mu)$ reduces to
\begin{equation} \label{for E 523}
E^{(5, (2-3))}(\mu) = 2 e^{-2} \chi_{ \acc{ \mu = \mathcal{L} |_{[0, 1]} } } + \frac{\limLolone}2 \indicator{ \Lol = 1 }.
\end{equation}

In some cases, it will be useful to reformulate $E^{(p, q)}$ in terms of elements from
\begin{equation}\label{for defn nsetcde}
\nsetcde := \bigaccv{ \xi : (0, 1) \rightarrow [0, \infty) }{ \xi \textrm{ non-decreasing} }.
\end{equation}
The elements $\xi \in \nsetcde$ relate to $\mu \in \mathcal{P} ([0, \infty))$ by being the inverse of the cumulative distribution of $\mu$. To state this more precise, we use the following notion of pseudo-invertibility. Let $f : (a, b) \rightarrow (c, d)$ non-decreasing, then we call
\begin{equation*}
    f^{-1}(y) := \sup \accv{ x \in (a, b) }{ f(x) < y }
\end{equation*}
the \emph{pseudo-inverse} of $f$. By using the pseudo-inverse, we can denote the relation between~$\xi$ and $\mu$ by
\begin{equation} \label{for link mu xi}
    \xi = \bighaa{ x \mapsto \mu ([0, x]) }^{-1}, \hsf \mu = D \lrhaa{ \xi^{-1} },
\end{equation}
where $D$ stands for the distributional derivative. Later on, in Theorem \ref{thm link BV and P}, we derive the related metric on $\nsetcde$, which allows us to prove $\Gamma$-convergence of the energies either on $\mathcal{P} ([0, \infty))$ or $\nsetcde$.

Before writing out explicitly the components of
\begin{equation*}
    E^{(p, q)}(\xi) = \Eita{p}{q}(\xi) + \EF{p}{q}(\xi) + \EL{p}{q}(\xi),
\end{equation*}
for $\xi \in \nsetcde$, we note that the following equalities follow from \eqref{for link mu xi}
\begin{gather*}
    \max \supp \mu = \sup \xi, \hsf
    \norm{\rho}{\infty} = \frac 1{ \inf \xi' },
\end{gather*}
where the second equality only makes sense if $\xi \in W^{1,1}_\text{incr}$ and if $\mu$ is absolutely continuous with $\rho = d\mu / d\mathcal{L}$. Together with these inequalities, it is easy to see
% \label{pageref:list_of_limit_energies}
\begin{align*}
    \Eita{3}{(0-3)}(\xi) &:= \lrhaa{ \int_0^\infty V } \int_0^1 \frac{1}{\xi'} \\
    \Eita{4}{(0-3)}(\xi) &:= \tilde{c} \int_0^1 \bigghaa{ \sum_{k = 1}^\infty V \lrhaa{ \tilde{c} k \xi' } } \\
    \Eita{5}{(0-3)}(\xi) &:= \left\{
                     \begin{array}{ll}
                       2 e^{-2} \indicator{\Lol \leq 1}, & \hbox{if $\xi' \geq 1$ a.e.,} \\
                       \infty, & \hbox{otherwise,}
                     \end{array}
                   \right. \\
    \EF{(1-5)}{(0-1)}(\xi) &:= \int_0^1 \xi , \\
    \EF{(1-5)}{2}(\xi) &:= C^{(p)}(\Lol{}) \int_0^1 \xi , \\
    \EF{(1-5)}{3}(\xi) &:= 0, \\
    \EL{(1-5)}{(0-1)} (\xi) &:= 0, \\
    \EL{(1-5)}{(2-3)} (\xi) &:= \chi_{ \acc{ \sup \xi \leq 1 } },
\end{align*}
where $C^{(p)}(\Lol{})$ is the same constant as defined in \eqref{for defn intro C} (we have changed the second argument to $p$ for convenience).

In cases $p = 3,4,5$, it turns out to be convenient to use both descriptions of $E^{(p, q)}$. Since it will be clear from the context in this paper which of the two descriptions we use, we do not make a distinction notation-wise.

Just as in \eqref{for ene dimless extension intro} we can regard $E_n^{(p, q)}$ as
\begin{gather*}
E_n^{(p, q)} : \nsetcde \rightarrow \overline{\R{}}, \\
    E_n^{(p, q)} (\xi) := \left\{
                                \begin{array}{ll}
                                  E_n^{(p, q)} (\xi), & \displaystyle \textrm{if } \exists x \in \Omega_n : \xi = \xi_n \textrm{ in the sense of Definition \ref{defn discrete seqs}}, \\
                                  \infty, & \textrm{otherwise}.
                                \end{array}
                              \right.
\end{gather*}
Again, we will not make a notational distinction.

\subsection{Using density to construct recovery sequences} \label{sec diag arg}

Lemma \ref{lem diag arg} will serve as the backbone for the proof of Theorem \ref{thm intro}.

\begin{lem} \label{lem diag arg}

($\limsup$ inequality for a dense subset). Let $M$ be a metric space, $M_1 \subset M$ dense, $F_n, F : M \rightarrow \R{}$. If
\begin{enumerate}[(i)]
  \item \label{lem diag arg cond i} \ddollar{ \forall v \in M_1 \hsq \exists (u_n) \subset M : u_n \rightarrow v \textrm{ and } \limp F_n(u_n) \leq F(v) }, and
  \item \label{lem diag arg cond ii} \ddollar{ \forall u \in M \hsq \exists (v_n) \subset M_1 : v_n \rightarrow u \textrm{ and } \limp F(v_n) \leq F(u) },
\end{enumerate}
then \ddollar{ \forall u \in M \hsq \exists (u_n) \subset M : u_n \rightarrow u \textrm{ and } \limp F_n(u_n) \leq F(u) }.

\end{lem}

\begin{rem} \label{rem lem diag arg}

The proof of Lemma \ref{lem diag arg} is based on a diagonal argument. See e.g. [\nbcite{MichielPaper}, Proposition 6.2] for the proof of a similar statement. Minor, obvious adjustments to that proof are needed to prove Lemma \ref{lem diag arg}.

\end{rem}

The following Lemma turns out to be very useful in our application of Lemma \ref{lem diag arg}. It gives a sufficient condition for condition $(ii)$, which consists of easier subproblems. In a way, it shows that one can show condition $(ii)$ iteratively. Since the proof can be done by a straight-forward diagonal argument, we do not show it here.

\begin{lem} \label{lem diag arg alt cond ii}

(Alternative for Lemma \ref{lem diag arg}, condition $(ii)$).
Let $M$ be a metric space, $M_1 \subset M$ dense, $k \in \{2,3,\ldots\}$, $M_1 \subset M_2 \subset \ldots M_k := M$, and $F_n, F : M \rightarrow \R{}$. If
\begin{equation*}
    \forall j \in \I{k-1} \hsq \forall u \in M_{j+1} \hsq \exists (v_n) \subset M_j : v_n \rightarrow u \textrm{ and } \limp E(v_n) \leq E(u),
\end{equation*}
then condition $(ii)$ of Lemma \ref{lem diag arg} is satisfied.

\end{lem}

\subsection{Link between $\mathcal{P} ([0, \infty))$ and $\nsetcde$} \label{sec link between domains}

In~\eqref{for link mu xi} we have shown how elements from $\mathcal{P} ([0, \infty))$ relate to those of $\nsetcde$. Here, we like to give a topology on $\nsetcde$ for which $\Gamma$-convergence of $E_n^{(p, q)}$ to $E^{(p, q)}$ on $\nsetcde$ with respect to that topology is equivalent to $\Gamma$-convergence of $E_n^{(p, q)}$ to $E^{(p, q)}$ on $\mathcal{P} ([0, \infty))$ with respect to narrow convergence. This statement follows easily from Theorem \ref{thm link BV and P}. Before stating it, we need two definitions:

\begin{defn} \label{defn topologies}

(Topology on $\nsetcde$).
Let $\xi_n, \xi \in \nsetcde$. We say that $\xi_n \rightharpoonup \xi$ in $BV_{\textrm{loc}}(0, 1)$ if for all $\delta \in (0,1)$ we have that $\xi_n \rightarrow \xi$ in $L^1(0, 1 - \delta)$ and $D \xi_n \rightharpoonup D\xi$ in $\mathcal{P} ((0, 1-\delta))$ with respect to the narrow topology,
where $D$ is the distributional derivative.

\end{defn}

\begin{rem}

Our motivation for using $\delta \in (0,1)$ instead of just taking $\delta = 0$, is that~$\xi(s)$ may go to $\infty$ if $s \uparrow 1$. This happens when the related $\mu \in \mathcal{P} ([0, \infty))$ has unbounded support.

\end{rem}

\begin{defn} \label{defn discrete seqs}

(Embedding discrete wall density).
For a sequence of $(n+1)$-tuples denoted by $\bighaa{ (x^n_i)_{i=0}^n }_{n \in \N{}_+}$ that satisfies $x_0^n = 0$ and $x_{i-1}^n \leq x_i^n$ for all $n \in \N{}_+$ and for all~${i \in \I{n}}$, we define $(\mu_n) \subset \mathcal{P} ([0, \infty))$ and $(\xi_n) \subset W^{1, \infty}(0,1)$ by
\begin{subequations}\label{for defn mun xin}
\begin{gather}
    \mu_n := \frac1n \sum_{i = 1}^n \delta_{ x_{i}^n }, \label{for defn mun xin a} \\
\xi_n(s) := x_{i-1}^n + n \lrhaa{ x_{i}^n - x_{i-1}^n } \Bighaa{ s - \frac{i-1}n }, \hsf \textrm{for } s \in \Bighaa{ \frac{i-1}n, \frac{i}n }. \label{for defn mun xin b}
\end{gather}
\end{subequations}

\end{defn}

\begin{rem}

We have made the choice to exclude $x_0^n$ from the definition of $\mu_n$.

A useful interpretation of $\mu_n$ and $\xi_n$ is as follows. For a Borel set $A$ of $\nsetcde$, the fraction of dislocation walls in $A$ is given by $\mu_n (A)$. $\xi_n$ uses the property that the walls are ordered by their position. $\xi_n(i/n)$ is the position of wall $i$. All the intermediate values of $\xi_n$ are chosen to be convenient in the $\Gamma$-convergence proofs.

\end{rem}

\begin{thm} \label{thm link BV and P}

(Link between $\mu$ and $\xi$ \cite{b32}).
Let $\bighaa{ (x^n_i)_{i=0}^n }_{n \in \N{}_+}, (\mu_n), (\xi_n)$ as in Definition \ref{defn discrete seqs}. Then the following two statements are equivalent:
\begin{enumerate}[(i)]
  \item $\xi_n \rightharpoonup \xi$ in $BV_{\text{loc}}(0, 1)$ ,
  \item $\mu_n \rightharpoonup \mu$.
\end{enumerate}

Moreover, if any of the two statements above hold, and $\xi^{-1} \in W^{1,1}(0, \infty)$, then
\begin{equation*}
    \rho := \ga{\mu}{\mathcal{L}} = \lrhaa{ \xi^{-1} }'.
\end{equation*}

\end{thm}

\subsection{Properties of the ``F" and ``L" part of the energies} \label{sec ene F L}

The energies $\EnF{p}{q}$ and $\EnL{p}{q}$ have special structure. The related properties are useful in reducing the complexity of the proof of Theorem \ref{thm intro}. Here, we make these properties precise.

Let
\begin{equation}\label{for defn Cn}
    C_n^{(p)}(\Lol_n, \alpha_n) := \left\{
                            \begin{array}{ll}
                              \Lol_n, & \hbox{if $p = 1$,} \\
                              \Lol_n^2, & \hbox{if $p \in \acc{2,3,4}$,} \\
                              \displaystyle \exp \bighaa{ 2 \alpha_n \haa{1 - 1 / \Lol_n } }, & \hbox{if $p = 5$.}
                            \end{array}
                          \right.
\end{equation}
If $q = 3$, we have $C_n^{(p)}(\Lol_n, \alpha_n) \rightarrow 0$. If $q = 2$, we obtain $C_n^{(p)}(\Lol_n, \alpha_n) \rightarrow C^{(p)}(\Lol{})$. We will require $C^{(p)}(\Lol{})$ to be finite. This means that for $p = 1$, we have to impose $\limLolone < \infty$. Note that this is exactly what we require in Theorem \ref{thm intro}.

Without violating \eqref{for ene dimless extension intro}, we can regard $\EnF{p}{q}, \EnL{p}{q} : \mathcal{P} ([0, \infty)) \rightarrow \R{}$ as
\begin{subequations}\label{for EnF is EF}
\begin{align}
    \EnF{p}{(0 - 1)} (\mu) &:= \EF{p}{(0 - 1)}(\mu) = \intabx{0}{\infty}{x}{\mu(x)} \label{for EnF is EF a}\\
    \EnF{p}{(2 - 3)} (\mu) &:= C_n^{(p)}(\Lol_n, \alpha_n) \intabx{0}{\infty}{x}{\mu(x)} \label{for EnF is EF b}\\
    \EnL{p}{(2 - 3)} (\mu) &:= \EL{p}{(2 - 3)}(\mu) \label{for EnF is EF c}
\end{align}
\end{subequations}
The following proposition is now a straightforward consequence of the statements above:

\begin{prop} \label{prop ct conv EnF}
(Continuous convergence of the force term).
For any $p \in \I{5}$ and any $q \in \{0, \ldots, 3\}$ (except for the case $p = 5$, $q = 2$ and $\limLolone = \infty$ (see \eqref{for defn intro limLolone})),
\begin{equation}
\label{for cont conv EF}
    \EnF{p}{q} \rightarrow \EF{p}{q},
\end{equation}
where the convergence is as in Definition \ref{defn continuous convergence} (i.e. continuous convergence) on the space $\mathcal{P} ([0, 1])$ with respect to the narrow topology.

\end{prop}

\begin{rem}
Proposition \ref{prop ct conv EnF} basically allows us to decouple the force term from the energy in the proof of Theorem \ref{thm intro} whenever $q = 2,3$. This is mainly due to Theorem \ref{thm gamma conv stable cont pert}, but we need additional arguments because the energy is defined on $\mathcal{P} ([0, \infty))$ instead of $\mathcal{P} ([0, 1])$.

\end{rem}

Proposition \ref{prop ct conv EnF} does not always apply due to the restriction to $\mathcal{P} ([0, 1])$. In that case, we still have lower semicontinuity, which also holds for $\EL{p}{(2 - 3)}$:

\begin{prop} \label{prop lsc}
$\EF{p}{(0 - 1)}$ and $\EL{p}{(2 - 3)}$ are lower semicontinuous on $\mathcal{P} ([0, \infty))$ with respect to the narrow topology.

\end{prop}

\begin{proof}[Proof of Proposition~\ref{prop lsc}]
Both $\EF{p}{(0 - 1)}$ and $\EL{p}{(2 - 3)}$ can be written as integrals over lower semicontinuous functions that are bounded from below. Lower semicontinuity of $\EF{p}{(0 - 1)}$ and $\EL{p}{(2 - 3)}$ follows from the Portmanteau Theorem.

\end{proof}

\section{Proof of Theorem~\ref{thm intro}} \label{sec thm pf}

Theorem \ref{thm intro} consists of two statements; a compactness property and $\Gamma$-convergence of the energies. The first can be proved in a few lines, which we do next. After that, we continue with the proof of the $\Gamma$-convergence.

By the compactness property, we mean that if for some $(\mu_n) \subset \mathcal{P}([0, \infty))$ it holds that $E^{(p, q)}_n (\mu_n)$ is bounded, then $(\mu_n)$ is compact in the narrow topology. For $q = 0$ this is given by [\nbcite{b32}, Theorem 1]. For $q = 1$, we have $E^{(p, 1)}_n \geq E^{(p, 0)}_n$, so the compactness property follows easily form the case $q = 0$. For $q = 2,3$, we have $E^{(p, (2 - 3))}_n (\mu_n) \leq C$ implies $\EnL{p}{(2 - 3)} (\mu_n) = 0$, and hence $\supp \mu_n \subset [0, 1]$. This gives tightness of $(\mu_n)$, and by e.g. Prokhorov's Theorem the compactness property follows.

We prove $\Gamma$-convergence of the energies by establishing the two inequalities
\begin{subequations}
\label{for proof th}
\begin{alignat}2
\label{for proof th liminf}
&\text{for all } \mu_n \weakto \mu, & \liminf_{n\to\infty}  E_n^{(p, q)}(\mu_n) &\geq E^{(p, q)}(\mu), \\
&\text{for all }\mu \text{ there exists }\mu_n\weakto \mu \text{ such that}\quad & \limsup_{n\to\infty} E_n^{(p, q)}(\mu_n) &\leq E^{(p, q)}(\mu),
\label{for proof th limsup}
\end{alignat}
\end{subequations}
for all $p=1,\dots,5$ and $q=0,\dots,3$ (except for the case $(p,q) = (5,2)$ and $\limLolone = \infty$). Here~$\mu_n$ and $\mu$ are probability measures on $[0,\infty)$. Note that it is sufficient to prove~\eqref{for proof th limsup} for all $\mu$ with $E^{(p, q)}(\mu)<\infty$.

In these inequalities, $E_n^{(p, q)}$ and $E^{(p,q)}$ are sums of three terms
\[
E_n^{(p, q)} = \Enita{p}{q} + \EnF{p}{q} + \EnL{p}{q},\quad\text{and}\quad
E^{(p, q)} = \Eita{p}{q} + \EF{p}{q} + \EL{p}{q},
\]
which are given in the list starting on page~\pageref{pageref:for En} and in Tables~\ref{tab Eita} and~\ref{tab EF and EL}.
Since similar results were proved in~\cite{b32} for similar energies without the final (``L'') term, we will be using many results from~\cite{b32}. The following lemma lists them. It uses the following (sub)spaces:
\begin{align*}
\dssetaba &:= \Bigaccv{ \mu \in \nsetab }{ \supp \mu \textrm{ bounded, } \mu \ll \mathcal{L} \textrm{, and } \ga{\mu}{\mathcal{L}} \in L^\infty (0, \infty) } \\
\nsetcde &= \bigaccv{ \xi : (0, 1) \rightarrow [0, \infty) }{ \xi \textrm{ non-decreasing} } \\
\dssetda &:= W^{1,1}_\textrm{incr} (0, 1) \\
\dssetca &:= \bigaccv{ \xi \in \dssetda }{ \xi' \geq \eps \textrm{ for some } \eps > 0 } \\
\dsseteaold &:= \bigaccv{ \xi \in \dssetda }{ \xi \textrm{ piece-wise affine} }.
\end{align*}
The tilde on $\dsseteaold$ is due to us using another definition for $\dssetea$ in the proof of Theorem~\ref{thm intro}.

\begin{lem}[Results from \cite{b32}] \label{lem lucia}
\indent
\begin{enumerate}[(i)]
\item ($\liminf$ inequality). Let $\mu_n,\mu\in \mathcal P([0, \infty))$, and $\mu_n\weakto\mu$.  For all $p\in\{1,\dots,5\}$ and all $q\in\{0,\dots,3\}$, we have
\begin{equation}
\label{for lem b32 liminf}
\liminf_{n\to\infty} \Enita pq(\mu_n) \geq \Eita pq(\mu).
\end{equation}
In addition, for all $0\leq I< J \leq n$,
\begin{equation}
\label{for ineq sumsumV V}
\frac1n \sum_{k=1}^n \sum_{j=0}^{n-k} V(n\alpha_n(x_{j+k}-x_j)) \geq \frac1n (J-I) V\Bigl(n\alpha_n\frac{x_J-x_I}{J-I}\Bigr).
\end{equation}

\item ($\limsup$ inequality). \label{lem lucia limsup} Let $p\in\{1,\dots,5\}$, $q \in \acc{0, 1}$, $\mu \in \mathcal P([0, \infty))$. Then there exists $\mu_n\weakto\mu$ such that
\begin{equation}
\label{for lem b32 limsup}
\limsup_{n\to\infty} E_n^{(p, q)} (\mu_n) \leq E^{(p, q)} (\mu).
\end{equation}

\item (Condition \eqref{lem diag arg cond i} of Lemma \ref{lem diag arg}). \label{lem lucia limsup explicit rec seq} Fix $p \in \I{5}$. If $p \leq 2$, let $\mu \in \dssetaba$. If $p = 3,4$, let $\xi \in \dssetalga$; if $p=5$, let $\xi \in \dsseteaold$. Let
\begin{align}\label{for defn rec seqs L 1 2}
    x^{(p),n}_i &:= \inf \bigaccv{ x \in [0, \infty) }{ \mu([0, x]) \geq i/n }, \hsf \textrm{ for } p = 1,2, \\\label{for defn rec seqs L 3 4}
    x^{(p),n}_i &:= \xi \lrhaa{ \frac{i}n }, \hsf  \textrm{for } p = 3,4, \\\label{for defn rec seqs L 5}
    x^{(p),n}_i &:= (1 + \eps_n) \xi \lrhaa{ \frac{i}n }, \hsf  \textrm{for } p = 5,
\end{align}
for some sequence $\eps_n \downarrow 0$. Let $\bighaa{ \mu_n^{(p)} }$ be defined as in~\eqref{for defn mun xin a}. Then
\begin{align}
\label{for lem lucia limsup expl rec seq}
\limp E_n^{(p, 0)} \bighaa{ \mu_n^{(p)} } &\leq E^{(p, 0)} \bighaa{ \mu^{(p)} }, \\
\label{for lem lucia limsup expl rec seq ita}
\limp \Enita p0 \bighaa{ \mu_n^{(p)} } &\leq \Eita p0 \bighaa{ \mu^{(p)} },
\end{align}
where $\mu^{(p)} := \mu$ if $p \leq 2$, and $\mu^{(p)} := \lrhaa{ \xi^{-1} }'$ (as in~\eqref{for link mu xi}) else.
%we can also make this explicit in terms of \xi, but lets see which we need, or both

\item (Condition \eqref{lem diag arg cond ii} of Lemma \ref{lem diag arg}). \label{lem lucia limsup cond ii} Fix $p \in \I{5}$. If $p \leq 2$, let $M := \mathcal{P}([0, \infty))$ and $M_1 := \dssetaba$, otherwise let $M := \nsetcde$ and $M_1 := \dssetalga$. Then condition \eqref{lem diag arg cond ii} of Lemma \ref{lem diag arg} holds for $F := E^{(p, 0)}$.
\end{enumerate}
\end{lem}

We now continue with the two inequalities~\eqref{for proof th}.

\subsection{The liminf inequality~\eqref{for proof th liminf}}

In cases $q=0,1$ either the domain is $[0,\infty)$ ($q=0$) or after rescaling the right-hand bound converges to $+\infty$ ($q=1$). Therefore the domain restriction enforced by $\EnL pq$ becomes unimportant in the limit $n\to\infty$, and for all $p$ we can simply disregard it:
\begin{eqnarray*}
\liminf_{n\to\infty} E_n^{(p,q)}(\mu_n)
&\stackrel{\eqref{for lem b32 liminf}, \eqref{for cont conv EF}}\geq&
 \Eita pq (\mu) + \EF pq (\mu) + \liminf_{n\to\infty} \EnL pq(\mu) \\
&\geq& \Eita pq (\mu) + \EF pq(\mu)
= E^{(p,q)}(\mu),
\end{eqnarray*}
which proves~\eqref{for proof th liminf} for all $p$ and for $q=0,1$.

In cases $q=2,3$, where the rescaled domain is $[0,1]$, the functional $\EnL pq$ becomes important. When $q=2,3$, $\EnL pq$ is independent of $n$ (see \eqref{for EnF is EF c}) and lower semicontinuous with respect to the narrow convergence (see Proposition \ref{prop lsc}). We then calculate for $p\in\{1,\dots,4\}$ and $q=2,3$,
\begin{eqnarray*}
\liminf_{n\to\infty} E_n^{(p,q)}(\mu_n)
&\stackrel{\eqref{for lem b32 liminf}, \eqref{for cont conv EF}}\geq&
 \Eita pq (\mu) + \EF pq (\mu) + \liminf_{n\to\infty} \EL pq(\mu_n) \\
&=& \Eita pq (\mu) + \EF pq(\mu) + \EL pq (\mu)
= E^{(p,q)}(\mu).
\end{eqnarray*}
This proves~\eqref{for proof th liminf} for these cases.

Finally, we discuss the case $q=2,3$ and $p=5$. Here the boundedness of the domain and the exponential behaviour of the tails of $V$ create a behaviour that is different from that on unbounded domains. We calculate, for any $0\leq I<J\leq n$,
\begin{eqnarray}
\notag
\Enita 5q(\mu_n) &=& \frac{ \exp \bighaa{2 \bighaa{ \alpha_n - 1 } }} { n \alpha_n }
\sum_{k=1}^n \sum_{j = 0}^{n - k} V \bighaa{ n \alpha_n \bighaa{ x_{j+k}^n - x_j^n } }
 \\
&\stackrel{\eqref{for ineq sumsumV V}}{\geq}&
  \frac{ \exp \bighaa{2 \bighaa{ \alpha_n - 1 } }} {\alpha_n}
    \frac1n (J-I) V\Bigl(n\alpha_n\frac{x^n_J-x^n_I}{J-I}\Bigr).
\label{ineq:Enita5}
\end{eqnarray}
Taking $I=0$ and $J=n$ in this expression, we find that
\begin{align}
\Enita 5q(\mu_n) &\geq  \frac{ \exp \bighaa{2 \bighaa{ \alpha_n - 1 } }} {\alpha_n}
    V(\alpha_n(x^n_n-x^n_0)) \notag\\
&\geq \frac{ \exp \bighaa{2 \bighaa{ \alpha_n - 1 } }} {\alpha_n}
    V(\alpha_n) && \text{since }x_n^n\leq 1 \notag\\
&= 2e^{-2} + \mathcal{O} (e^{-2\alpha_n})  &&\text{by Prop.~\ref{properties V}\eqref{prop asymp V large t}}. \label{thm intro pf 1}
\end{align}
Therefore
\begin{align*}
\liminf_{n\to\infty} E_n^{(5,q)}(\mu_n)  &\geq 2e^{-2} + \liminf_{n\to\infty} \bigl[ \EnF 5q(\mu_n) + \EnL5q(\mu_n)\bigr]\geq 2e^{-2}.
\end{align*}

In order to show that $\liminf_{n\to\infty} E_n^{(5,q)}(\mu_n) \geq E^{(5,q)}(\mu)$, we still need to show that $\liminf_{n\to\infty} E_n^{(5,q)}(\mu_n) = \infty$ whenever $\mu \not= \mathcal L |_{[0,1]}$. If $\supp \mu \nsubseteq [0, 1]$, we have that ${\EnL 5q (\mu_n) = \infty}$ by \eqref{for EnF is EF c} and Proposition \ref{prop lsc}. If $\supp \mu \subset [0, 1]$ and $\mu \not= \mathcal L |_{[0,1]}$, there exists an interval $(a,b)\subset \R{}$ such that $\delta := (b-a)^{-1}\mu\bigl( (a,b) \bigr)>1$. Define $I_n$ and~$J_n$ by
\[
x^n_{I_n} = \min_i \bigaccv{ x_i^n }{ x_i^n > a } \qquad\text{and}\qquad
x^n_{J_n} = \max_i \bigaccv{ x_i^n }{ x_i^n < b }.
\]
Using Prokhorov's characterization of narrow convergence, we calculate
\[
\limsup_{n\to\infty} \delta (x^n_{J_n}-x^n_{I_n})
\leq \delta(b-a)
= \mu\bigl((a,b)\bigr) \leq \liminf_{n\to\infty}\mu_n\big((a,b)\bigr)
= \liminf_{n\to\infty} \frac1n (J_n - I_n + 1),
\]
and therefore
\[
\limsup_{n\to\infty} \;n\frac{x^n_{J_n}-x^n_{I_n}}{J_n-I_n}\leq \frac1\delta.
\]
Continuing from~\eqref{ineq:Enita5} we then find
\begin{align}\label{thm intro pf 2}
\Enita 5q(\mu_n)&\geq \frac{ \exp \bighaa{2 \bighaa{ \alpha_n - 1 } }} {\alpha_n}
    \frac1n (J_n-I_n) V(\alpha_n\delta^{-1})\notag\\
&\geq 2e^{-2} (b-a) \exp\bigl[2\alpha_n(1-\delta^{-1})\bigr] \bigl(1+ \mathcal{O} (e^{-2\alpha_n/\delta})\bigr).
\end{align}
This converges to $+\infty$ since $\delta>1$.

\subsection{The limsup inequality~\eqref{for proof th limsup}}

\textbf{The case $q=0$.} When $q = 0$, \eqref{for proof th limsup} is given by Lemma \ref{lem lucia}.\eqref{lem lucia limsup}. However, for the specific case $p = 5$, we present an alternative proof here. The proof is easier and more explicit than the proof as given in \cite{b32}. Moreover, the arguments in the following proof are easier to extend to the cases in which $q \neq 0$.

We conclude \eqref{for proof th limsup} from Lemma \ref{lem diag arg} after showing that its two conditions are satisfied. We use Lemma \ref{lem diag arg} with the subset
\begin{equation*}
    \dssetea := \bigaccv{ \xi \in \dssetdb }{ \inf \xi' > 1 }.
\end{equation*}

\emph{Condition \eqref{lem diag arg cond i}.} Let $\xi \in \dssetea$. We construct $\xi_n$ by using linear interpolation (see \eqref{for defn mun xin b}) with $x_i^n := \xi(i/n)$. Observe that for any $i,j \in \{0, \ldots, n\}$ with $i > j$, we have the estimate
\begin{equation}\label{thm intro pf 5}
    \lrhaa{ x^n_i - x^n_j } = \bigabs{ \xi(i/n) - \xi(j/n) } = \int_{j/n}^{i/n} \xi' \geq (\inf \xi') \frac{i - j}n.
\end{equation}
Let $m := \inf \xi' > 1$. We calculate
\begin{align}
&\sum_{k=1}^n \sum_{j = 0}^{n - k} V \bighaa{ n \alpha_n \bighaa{ x_{j+k}^n - x_j^n } } && \notag\\
\leq \: &\sum_{k=1}^n \sum_{j = 0}^{n - k} V \lrhaa{ n \alpha_n m \frac kn } && \text{by \eqref{thm intro pf 5} and $V$ decreasing} \notag\\
= \: &\sum_{k=1}^n (n - k + 1) 2 m k \alpha_n e^{ -2 m k \alpha_n } \bighaa{ 1 + \mathcal{O} (e^{ -2 m k \alpha_n }) } && \text{by Prop. \ref{properties V}\eqref{prop asymp V large t} } \notag\\
\leq \: &2 m n \alpha_n \bighaa{ 1 + \mathcal{O} (e^{ -2 m \alpha_n }) } \sum_{k=1}^n k e^{ -2 m k \alpha_n }  && \notag\\
= \: &2 m n \alpha_n \bighaa{ 1 + \mathcal{O} (e^{ -2 m \alpha_n }) } e^{ -2 m \alpha_n } \bighaa{ 1 + \mathcal{O} (e^{ -2 m \alpha_n }) }, && \label{thm intro pf 4}
\end{align}
from which it follows that
\begin{align*}
\Enita 50 (\xi_n) &= \frac{ \exp \bighaa{2 \bighaa{ \alpha_n - 1 } } }{ n \alpha_n } \sum_{k=1}^n \sum_{j = 0}^{n - k} V \bighaa{ n \alpha_n \bighaa{ x_{j+k}^n - x_j^n } } \\
&\leq \frac {2 m}{e^2} e^{- 2 \alpha_n (m - 1)} \bighaa{ 1 + \mathcal{O} \haa{ e^{-2 m \alpha_n} } } \rightarrow 0.
\end{align*}
It remains to show that the limsup also holds for the force term. As $\xi(1) < \infty$, it is allowed to use Proposition \ref{prop ct conv EnF} to conclude that $\EnF 50 \rightarrow \EF 50$ continuously. % We gebruiken hier ook de mu xi link

\emph{Condition \eqref{lem diag arg cond ii}.} By Lemma \ref{lem diag arg alt cond ii} it is enough to show that the following two inclusions are energy dense:
\begin{equation} \label{thm particular case pf 4}
\dssetea \subset \dssetda \subset \nsetcde \hsf \textrm{with respect to } E^{(5, 0)} .
\end{equation}
Energy density of the second inclusion follows from Theorem \ref{thm T2.4}. The first inequality is easy to prove: take $\xi \in \dssetda$ with $E^{(5, 0)}(\xi) < \infty$. This implies $\inf \xi' \geq 1$. Hence $\xi_n := (t \mapsto \xi(t) + t/n) \in \dssetea$, $\Eita 50 (\xi_n) = 0 = \Eita 50 (\xi)$, and $\EF 50 (\xi_n) \rightarrow \EF 50 (\xi)$. This completes the proof for case $(p, q) = (5, 0)$.

\textbf{Case $q=1$.} We continue with case $q = 1$ for any $p$. The expressions for $E_n^{(p, 1)}$ and~$E^{(p, 1)}$ are very similar to those from case $q = 0$, because both the interaction and force term of the related energies are the same. However, the presence of the second barrier may make the recovery sequence as given implicitly by Lemma \ref{lem lucia}\eqref{lem lucia limsup} not applicable. Our strategy to solve this issue is to take the explicitly given recovery sequence (only for special choices for $\mu$ (see \eqref{for defn rec seqs L 1 2} -- \eqref{for defn rec seqs L 5})), show that these recovery sequences also work in case $q = 1$, and extrapolate these results to general $\mu \in \mathcal{P}([0, \infty))$ via Lemma \ref{lem diag arg}.

If $p \leq 2$, let $\mu \in \dssetaba$, otherwise let $\xi \in \dssetalga$ and define $\mu := \lrhaa{ \xi^{-1} }'$ (as in~\eqref{for link mu xi}). Let $\mu_n$ as in Lemma \ref{lem lucia}\eqref{lem lucia limsup explicit rec seq}. By using this Lemma and $\max \supp \mu_n \leq C$, we obtain
\begin{align*}
    \limp E^{(p, 1)}_n \haa{ \mu_n }
&\leq \limp E^{(p, 0)}_n \haa{ \mu_n } + \limp \EnL p1 ( \mu_n )\\
&\leq E^{(p, 0)} \bighaa{ \mu } = E^{(p, 1)} \bighaa{ \mu }.
\end{align*}
Together with Lemma \ref{lem lucia}\eqref{lem lucia limsup cond ii} and the observation that $E^{(p, 1)} = E^{(p, 0)}$, we see that the two conditions from Lemma \ref{lem diag arg} are satisfied, from which \eqref{for proof th limsup} follows.

\textbf{Cases $q = 2, 3$.} Here we separate the proof for $p \leq 4$ and $p = 5$. In the latter case, we have that $E^{(5, q)} \haa{ \mu }$ can only be finite if $\mu = \mathcal{L} |_{[0,1]}$, for which the proof requires a different argument.

We start with $p \leq 4$. Note that the energies are much alike for $q = 2,3$: we have $E_n^{(p, 2)} = E_n^{(p, 3)}$ and $E^{(p, 2)} = E^{(p, 3)} + \EF{p}{2}$. Hence we take $q \in \acc{2,3}$ arbitrary.

Since we can restrict ourselves to those $\mu \in \mathcal{P}([0, \infty))$ for which $\EL pq (\mu)$ is finite, we can assume that $\supp \mu \subset [0, 1]$ and $\EL pq (\mu) = 0$. We prove \eqref{for proof th limsup} by applying Lemma \ref{lem diag arg} to the following spaces:
\begin{subequations}\label{thm intro pf 3}
\begin{align}
\nsetabb &:= \bigaccv{ \mu \in \nsetab }{ \supp \mu \subset [0, 1] }, \label{thm intro pf 3 a}\\
\dssetabb &:= \bigaccv{ \mu \in \dssetaba }{ \supp \mu \subset [0, 1] }, \label{thm intro pf 3 b}\\
\nsetcdeb &:= \bigaccv{ \xi \in \nsetcde }{ \sup \xi \leq 1 }, \label{thm intro pf 3 c}\\
\dssetalgb &:= \bigaccv{ \xi \in \dssetalga }{ \sup \xi \leq 1 } , \hsf \textrm{for } p = 3,4. \label{thm intro pf 3 d}
\end{align}
\end{subequations}
It remains to show that the two conditions of Lemma \ref{lem diag arg} are satisfied:

\emph{Condition \eqref{lem diag arg cond i}.} Let $p = 1,2$ and $\mu \in \dssetabb$. Let $\mu_n$ as in Lemma \ref{lem lucia}\eqref{lem lucia limsup explicit rec seq}. Observe that $\supp \mu_n \subset [0, 1]$, so
\begin{eqnarray*}
  && \: \limp E_n^{(p, q)} (\mu_n)  \\
&\leq& \: \limp \Enita pq (\mu_n) + \limp \EnF pq (\mu_n) + \limp \EnL pq (\mu_n)  \\
  &\stackrel{ \eqref{for lem lucia limsup expl rec seq ita}, \eqref{for cont conv EF} }{\leq}& \: \Eita pq (\mu) + \EF pq (\mu) = E^{(p, q)} (\mu)
\end{eqnarray*}
For $p = 3,4$, we can repeat the same argument for $\xi \in \dssetalgb$.

\emph{Condition \eqref{lem diag arg cond ii}.} As $\EF pq$ is continuous on $\nsetabb$, it is sufficient to prove condition \eqref{lem diag arg cond ii} for the interaction part. If $p = 1,2$, this condition is given by Lemma \ref{lem ene dens 12}\eqref{lem ene dens 12 part ii}. For $p = 3,4$, we use Lemma \ref{lem diag arg alt cond ii} to argue that we can split the proof by showing separately that the following three inclusions are lower energy dense:
\begin{align*}
    \dssetcb \stackrel{(\textrm{a})}{\subset} \dssetdb \stackrel{(\textrm{b})}{\subset} \nsetcdeb \hsf &\textrm{with respect to } \Eita{3}{q}, \\
\dssetdb \stackrel{(\textrm{c})}{\subset} \nsetcdeb \hsf &\textrm{with respect to } \Eita{4}{q}.
\end{align*}
Energy density of inclusions (b) and (c) follows from Theorem \ref{thm T2.4}. For inclusion (a), let $\xi \in \dssetcb$ and $\eps_n \downarrow 0$. Take $\xi_n(t) := (1 + \eps_n)^{-1} \bighaa{ \xi(t) + \eps_n t }$. Then
\begin{equation*}
\xi_n' = \frac{ \xi' + \eps_n }{ 1 + \eps_n } \geq \frac{ \eps_n }{ 1 + \eps_n } > 0, \hsf \sup \xi_n \leq 1,
\end{equation*}
and hence $(\xi_n) \subset \dssetcb$. Obviously, $\xi_n \rightharpoonup \xi \textrm{ in } BV_{\textrm{loc}}(0,1)$, and
\begin{align*}
    \limp \Eita 3q (\xi_n) &= \lrnorm{V}{L^1(0, \infty)} \limp \int_0^1 \frac{1}{\xi_n'} \\
&\leq \lrnorm{V}{L^1(0, \infty)} \limp \int_0^1 \frac{1 + \eps_n}{\xi'} = \Eita 3q (\xi).
\end{align*}

Let $p = 5$. As said before, we only have to regard $\mu = \mathcal{L} |_{[0,1]}$, because $E^{(5, q)} (\mu)$ is infinite for any other $\mu$. We take the sequence $\mu_n$ related to $x_i^n = i/n$. Clearly $\mu_n \weakto \mathcal{L} |_{[0,1]}$. We prove \eqref{for proof th limsup} by explicitly calculating the $\limsup$ of all three parts of the energy.

Obviously, $\EnL 5q (\mu_n) = 0$, and
\begin{equation*}
\EnF 5q (\mu_n) \xrightarrow{\eqref{for cont conv EF}} \EF 5q \bighaa{ \mathcal{L} |_{[0,1]} } = \frac {\limLolone}2,
\end{equation*}
where $\limLolone$ is defined in \eqref{for defn intro limLolone}. It is exactly here that we need the condition $\limLolone < \infty$ as imposed in Theorem \ref{thm intro}, because we need $E^{(5, q)} \bighaa{ \mathcal{L} |_{[0,1]} }$ to be finite to obtain a non-trivial limit energy.

For the limsup of $\Enita 5q (\mu_n)$, we use estimate \eqref{thm intro pf 4} with $m = 1$ to obtain
\begin{align*}
\Enita 5q (\mu_n) &= \frac{ \exp \bighaa{2 \bighaa{ \alpha_n - 1 } } }{ n \alpha_n } \sum_{k=1}^n \sum_{j = 0}^{n - k} V \bighaa{ n \alpha_n \bighaa{ x_{j+k}^n - x_j^n } } \\
&\leq \frac 2{e^2} \bighaa{ 1 + \mathcal{O} \haa{ e^{-2 \alpha_n} } } \rightarrow \frac 2{e^2} .
\end{align*}
By gathering the results above, we obtain
\begin{align*}
\limp E_n^{(5, q)} (\mu_n) &\leq \limp \Enita 5q (\mu_n) + \limp \EnF 5q (\mu_n) + \limp \EnL 5q (\mu_n) \\
&\leq \frac 2{e^2} + \frac {\limLolone}2 = E^{(5, q)} \lrhaa{ \mathcal{L} |_{[0,1]} }.
\end{align*}

\section{Further results and applications} \label{sec further results}

Although the proof of Theorem \ref{thm intro} is complete, we still need to treat the special case (i.e.~$p = 5$, $q = 2$ and $\limLolone = \infty$ (see \eqref{for defn intro limLolone})) which is not covered by Theorem \ref{thm intro}. Furthermore, we show that $E^{(p, q)}$ has a unique minimizer, which is, moreover, the limit of the sequence of minimizers of $E_n^{(p, q)}$.

\subsection{The particular case $p = 5$, $q = 2$ and $\limLolone = \infty$}

As mentioned in the introduction, the term coming from the finite domain is negligible with respect to the force term if $p = 5$ and $\limLolone = \infty$. By considering the scaling of $\mathcal{E}$ as given by $E^{(5,2)}_n$, the only candidate for the $\Gamma$-limit would be $\infty$ (we do not prove this), which means that $E^{(5,2)}_n$ does not contain information in the limit. This is not unexpected, because this scaling of $\mathcal{E}$ is based on balancing the interaction term with the term coming from the finite domain. Here, we consider the scaling coming from balancing the interaction term with the force term (see \eqref{for ene dimful intro}). Let $\Enh := E_n^{(5,1)}$. Because we only consider the specific case $p = 5$ and $\limLolone = \infty$ in this section, we do not incorporate it in the notation of $\Enh$, nor in its $\Gamma$-limit $\Eh$, which is defined by
\begin{gather}
    \Eh : \mathcal{P}([0, \infty)) \rightarrow \R{} \notag\\
\begin{align}
\Eh (\mu) &:= \Ehita (\mu) + \EhF (\mu) + \EhL (\mu) \label{for defn Eh}\\
\Ehita (\mu) &:= \Eita{5}{0} (\mu) \notag\\
\EhF (\mu) &:= \EF{5}{1} (\mu) \notag\\
\EhL (\mu) &:= \left\{
                     \begin{array}{ll}
                       0, & \hbox{if $\supp \mu \subset [0, \Lol]$,} \\
                       \infty, & \hbox{otherwise.}
                     \end{array}
                   \right. \notag
\end{align}
\end{gather}
We emphasize that $\Ehita ( \mathcal{L} |_{(0, 1)} ) = 0$, even when $\Lol = 1$. Just as before, we regard $\Eh$ as a mapping from $\nsetcde$ to $\R{}$ whenever that is more convenient.

\begin{thm} \label{thm particular case}
(Convergence of the energy; particular case).
Let $p = 5$, $q = 2$ and $\limLolone = \infty$, and consider $\mathcal{P} ([0, \infty))$ equipped with the narrow topology. If $(\mu_n) \subset \mathcal{P}([0, \infty))$ is such that $\Enh (\mu_n)$ is bounded, then $(\mu_n)$ is compact. Moreover, $\Enh$ $\Gamma$-converges to $\Eh$.

\end{thm}

\begin{proof}[Proof of Theorem~\ref{thm particular case}]
The proof is similar to the proof of Theorem~\ref{thm intro}. In fact, the proof for the compactness statement is the same, so we do not repeat it here. The proof for the $\Gamma$-convergence again consists of proving the following two inequalities:
\begin{subequations}
\label{thm particular case pf 1}
\begin{alignat}2
\label{thm particular case pf 2}
&\text{for all } \mu_n \weakto \mu, & \liminf_{n\to\infty}  \Enh (\mu_n) &\geq \Eh (\mu), \\
&\text{for all }\mu \text{ there exists }\mu_n\weakto \mu \text{ such that}\quad & \limsup_{n\to\infty} \Enh (\mu_n) &\leq \Eh (\mu),
\label{thm particular case pf 3}
\end{alignat}
\end{subequations}

For \eqref{thm particular case pf 2}, note that by \eqref{for lem b32 liminf} we have~$\limf \Enhita (\mu_n) \geq \Ehita (\mu)$, and by \eqref{for EnF is EF} and Proposition \ref{prop lsc}, we have
\[
\limf \EnhF (\mu_n) \geq \EhF (\mu), \hsf \limf \EnhL (\mu_n) \geq \EhL (\mu).
\]
Together these prove~\eqref{thm particular case pf 2}.

We prove \eqref{thm particular case pf 3} separately for $\Lol > 1$ and $\Lol = 1$. In the first case, we use Theorem \ref{thm link BV and P} to prove \eqref{thm particular case pf 3} for non-decreasing functions $\xi$. We can restrict ourselves to proving \eqref{thm particular case pf 3} only for $\xi \in \nsetcdeblol$; for these $\xi$, $\EhL (\xi) = 0$. The subscript in the notation for $\nsetcde$ refers to the upper bound for $\sup \xi$, just as it did in the spaces defined by \eqref{thm intro pf 3}. By Proposition~\ref{prop ct conv EnF}, this upper bound on $\xi$ implies that the force term is a continuous perturbation to $\Enh$, so by Theorem \ref{thm gamma conv stable cont pert} it is enough to prove
\begin{gather}
    \text{for all }\xi \in \nsetcdeblol \text{ there exists }\xi_n \weakto \xi \text{ in $BV_{\text{loc}}(0, 1)$ such that} \notag\\
 \limsup_{n\to\infty} \Enhita (\xi_n) + \EnhL (\xi_n) \leq \Ehita (\xi), \label{thm particular case pf 5}
\end{gather}

We prove \eqref{thm particular case pf 5} by applying Lemma \ref{lem diag arg} with the subset
\[
\dsseteclol := \bigaccv{ \xi \in \dsseteblol }{ \xi(1) < \Lol }.
\]
This requires its two conditions to be satisfied:

\emph{Condition \eqref{lem diag arg cond i}.} Let $\xi \in \dsseteclol$ and take $(\xi_n) \subset \nsetcdeblol$ as defined by \eqref{for defn rec seqs L 5}. Note that $\sup \xi_n \rightarrow \sup \xi < \Lol $, which together with $\Lol_n \rightarrow \Lol $ implies that indeed $(\xi_n) \subset \nsetcdeblol$ for all $n$ large enough. Furthermore, we have $\sup \xi_n \leq \Lol_n$ for all $n$ large enough, which implies $\EnhL (\xi_n) = 0$. Hence
\begin{equation*}
    \limp \Enhita (\xi_n) + \EnhL (\xi_n) = \limp \Enita 50 (\xi_n) \stackrel{\eqref{for lem lucia limsup expl rec seq ita}}{\leq} \Eita 50 (\xi).
\end{equation*}

\emph{Condition \eqref{lem diag arg cond ii}.} By Lemma \ref{lem diag arg alt cond ii} it is enough to show that the following three inclusions are energy dense:
\begin{equation} \label{thm particular case pf 4}
\dsseteclol \subset
\dssetdblol \subset
\nsetcdeblol \hsf \textrm{with respect to } \Ehita .
\end{equation}
Energy density of the second inclusion follows from Theorem \ref{thm T2.4}. To show the first inclusion, we take $\xi \in \dssetdblol$. This implies that $\xi \in W^{1,1}_{\text{incr}}$, $\xi (1) \leq \Lol$ and $\inf \xi' \geq 1$. It is enough to construct $\xi_n \weakto \xi$ in $BV_{\text{loc}}$ such that $\xi_n \in W^{1,1}_{\text{incr}}$, $\xi_n (1) < \Lol$ and $\inf \xi_n' > 1$, because then $(\xi_n) \subset \dsseteclol$ and $\Enhita (\xi_n) + \EnhL (\xi_n) = 0$. It is easy to see that $\xi_n$ as defined by
\begin{equation*}
\xi_n(t) := \frac{1 - \eps_n}{1 + \eps_n} (\xi(t) + \eps_n \Lol t) + \eps_n t
\end{equation*}
for some $\eps_n \downarrow 0$, satisfies all these requirements. Note that the strict inequalities in the requirements for $\xi_n$ are obtained solely by using $\Lol > 1$. This completes the proof for~\eqref{thm particular case pf 3} under the assumption that $\Lol > 1$.

We now turn to the case $\Lol = 1$. As the following proof is similar to the proof of Theorem \ref{thm intro} in case $p = 5$ and $q = 2,3$ (see page \pageref{thm intro pf 4}), we do this in terms of measures instead of using non-decreasing functions. Again, we have that $\Eh (\mu)$ can only be finite when $\mu = \mathcal{L} |_{(0,1)}$, but now we take $\mu_n$ as defined by $x_i^n := \Lol_n i/n$. This is to ensure that $\EnhL (\mu_n) = 0$. Clearly $\EnhF(\mu_n) \rightarrow 1/2 = \EhF \bighaa{ \mathcal{L} |_{(0,1)} }$, so it is only left to prove that $\Enhita (\mu_n) \rightarrow 0$. Due to $x_i^n = \Lol_n i/n$, we get
\begin{equation*}
\sum_{k=1}^n \sum_{j = 0}^{n - k} V \bighaa{ n \ahat_n \bighaa{ x_{j+k}^n - x_j^n } } \stackrel{ \eqref{thm intro pf 4} }{\leq} 2 n \ahat_n \Lol_n e^{-2 \Lol_n \ahat_n} + \mathcal{O} \lrhaa{ e^{-4 \Lol_n \ahat_n} },
\end{equation*}
and hence
\begin{eqnarray*}
\Enhita (\mu_n) &=& \frac{ \exp \bighaa{2 \bighaa{ \ahat_n - 1 } } }{ n \ahat_n } \sum_{k=1}^n \sum_{j = 0}^{n - k} V \bighaa{ n \ahat_n \bighaa{ x_{j+k}^n - x_j^n } } \\
&\leq& \frac 2{e^2} \Lol_n \bighaa{ e^{2 \ahat_n (1 - \Lol_n)} + \mathcal{O} \lrhaa{ e^{-2 \ahat_n} } } \\
&\stackrel{ \eqref{for defn intro alpha} }{\leq}& \frac 2{e^2} \Lol_n \bighaa{ e^{2 \alpha_n (1/\Lol_n - 1)} + \mathcal{O} \lrhaa{ e^{-2 \ahat_n} } } \rightarrow 0,
\end{eqnarray*}
in which the convergence to $0$ follows from $\limLolone = \infty$.

\end{proof}

\subsection{Existence, uniqueness and convergence of minimizers}

\begin{thm} \label{thm main appl}
(Existence and uniqueness of minimizers). Let $p \in \I{5}, q \in \{0, \ldots, 3\}$. The minimization problem
\begin{equation*}
    \min_{\mu \in \mathcal{P}([0, \infty))} E^{(p, q)} (\mu)
\end{equation*}
has a unique minimizer. The energy $\Eh$ (as defined in \eqref{for defn Eh}) has a unique minimizer as well.

\end{thm}

\begin{proof}[Proof of Theorem~\ref{thm main appl}]
For case $q = 0$, this has been proved in (\cite{b32}, Theorem 2). Because our proofs for $q = 1,2,3$ are similar to that proof, we state the intermediate results of that proof first.

To show existence, take a minimizing sequence $(\mu_m)_{m\in \N{}}$. Since for each of the limit energies  either $\EF pq$ or $\EL pq$ is non-vanishing, $(\mu_m)$ is tight, and therefore narrowly compact. Since each of the terms in the limiting energies is narrowly lower semicontinuous, existence follows.

To show uniqueness, note that $E^{(\mathrm F)}$ and $E^{(\mathrm L)}$ are convex, both in the classical sense, i.e. in the additive structure on $\mathcal P([0,\infty))$, and displacement convex. In~\cite{b32} it was shown that~$\Eita p {(0-3)}$ is strictly convex in the classical sense for $p=1,2,3$ and strictly geodesically convex for $p=4$. For all $p\leq4$, therefore, $E^{(p,(0-3))}$ is strictly convex in some sense and therefore has exactly one minimizer.
If $p = 5$, it is obvious from \eqref{for E 523} that $\mathcal{L} |_{[0, 1]}$ is the unique minimizer of $\EF{5}{q}$ when $\limLolone < \infty$. If $\beta = \infty$ (the case of Theorem~\ref{thm particular case}), the limit energy is given by $\Eh (\mu) = E^{(5, 0)} (\mu) + \chi_{ \acc{ \supp \mu \subset [0, \Lol] } }$, for which $\mathcal{L} |_{[0, 1]}$ is the unique minimizer.
\end{proof}

\begin{cor} \label{cor main thm}
(Convergence of minimizers).
For each $n \in \N{}_+$, let $\mu_n^\ast$ and $\mu^\ast$ be the minimizers of respectively $E_n^{(p,q)}$ and $E^{(p,q)}$ (or $\Enh$ and $\Eh$ whenever $p = 5$ and $\beta = \infty$). Then $\mu_n^\ast \rightharpoonup \mu^\ast$.

\end{cor}

\begin{proof}[Proof of Corollary~\ref{cor main thm}]
The proof is the same for $p = 5$ and $\beta = \infty$ as for the other cases. Hence we restrict ourselves to the other cases, and so we use the energies $E_n^{(p,q)}$ and $E^{(p,q)}$.

By Theorem~\ref{thm intro}, the sequence $(\mu_n^\ast)$ is narrowly compact, and converges along a subsequence to a limit $\mu$. By standard properties of $\Gamma$-convergence, $\mu$ is a minimizer of $E^{(p,q)}$. Since minimizers of $E^{(p,q)}$ are unique by Theorem~\ref{thm main appl}, the whole sequence converges.
\end{proof}

\subsection{Rescaling $E_n^{(5, (2-3))}$}

As mentioned in Section \ref{sec:discussion-intro}, the $\Gamma$-limit of $E_n^{(5, (2-3))}$ is unsatisfactory, because it only contains information about the unique minimizer. One way to keep more information in the limit, is to consider a logarithmic scaling. More precisely, we define
\begin{equation*}
    \logEn (\mu_n) := \frac 1{2 \alpha_n} \log E_n^{(5, (2-3))} (\mu_n)
\end{equation*}
and show that it $\Gamma$-converges to $\logE$, which is given by
\begin{equation*}
    \logE (\mu) = \left\{
                    \begin{array}{ll}
                      1 - \frac 1{\maxmu }, & \hbox{if $\supp \mu \subset [0, 1]$,} \\
                      \infty , & \hbox{otherwise,}
                    \end{array}
                  \right.
\end{equation*}
where
\begin{equation*}
     \maxmu{} := \sup_{a < b} \frac{ \mu\bigl( (a,b) \bigr) }{ b - a }.
\end{equation*}
We can also express $\logE$ in terms of non-decreasing functions as
\begin{equation*}
    \logE (\xi) = \left\{
                    \begin{array}{ll}
                      1 - m_\xi, & \hbox{if $\sup \xi \leq 1$,} \\
                      \infty , & \hbox{otherwise,}
                    \end{array}
                  \right.
\end{equation*}
where
\begin{equation*}
     m_\xi := \inf_{a < b} \frac{ D \xi \bighaa{(a, b)} }{ b - a },
\end{equation*}
and $D \xi$ is the distributional derivative of $\xi$.

\begin{thm} \label{thm log}

($\Gamma$-Convergence of the logarithm of the energy).
Let $p = 5$, $q = 2$ and $\limLolone = \infty$, and consider $\mathcal{P} ([0, \infty))$ equipped with the narrow topology. If $(\mu_n) \subset \mathcal{P}([0, \infty))$ is such that $\Enh (\mu_n)$ is bounded, then $(\mu_n)$ is precompact. Moreover, $\Enh$ $\Gamma$-converges to~$\Eh$.

\end{thm}

\begin{proof}[Proof of Theorem~\ref{thm log}]

The structure of the proof is similar to the $\Gamma$-convergence proof of $E_n^{(5, (2-3))}$. Compactness follows from the same argument as used for showing compactness for $E_n^{(5, (2-3))}$, because we still require for any fixed $n$ that $\EnL 5{(2-3)} (\mu_n) < \infty$ in order for $\logEn (\mu_n)$ to be finite.

To show the liminf inequality, we separate three cases: $\mu = \mathcal{L} |_{(0, 1)}$, $\supp \mu \nsubseteq [0,1]$, and all other $\mu \in \mathcal{P} ([0, \infty))$.

If $\mu = \mathcal{L} |_{(0, 1)}$, we see from \eqref{thm intro pf 1} for any $\mu_n \weakto \mu$ that
\begin{equation*}
\limf \logEn (\mu_n) \geq \limf \frac 1{2 \alpha_n} \log \lrhaa{ 2e^{-2} + \mathcal{O} (e^{-2\alpha_n}) } = 0.
\end{equation*}

If $\supp \mu \nsubseteq [0, 1]$, it follows from lower semi-continuity (see \eqref{for EnF is EF c} and Proposition~\ref{prop lsc}) that $\EnL 5{(2-3)}(\mu_n) = \infty$ for $n$ large enough for any $\mu_n \weakto \mu$, so that $\logEn (\mu_n) = \infty$ as well.

If $\supp \mu \subset [0, 1]$ and $\mu \neq \mathcal{L} |_{(0, 1)}$, we have that $1 < \maxmu{}$. As we like to have explicit values for $a,b$ in the calculation below (rather than the supremum over them as in the definition of $\maxmu$), we fix $0 < \eps < \maxmu{} - 1$, and take $a_{\eps} < b_{\eps}$ such that $(b_{\eps} - a_{\eps})^{-1} \mu\bigl( (a_{\eps},b_{\eps}) \bigr) > \maxmu{} - \eps =: \maxmu{}_{\eps}$. We follow the same reasoning as for \eqref{thm intro pf 2} to find
\begin{align*}
\limf \logEn (\mu_n) &\geq \limf \frac 1{2 \alpha_n} \log \lrhaa{ 2 e^{-2} (b_{\eps} - a_{\eps}) \exp \bighaa{ 2 \alpha_n ( 1 - \maxmu{}_{\eps}^{-1} ) } \bighaa{ 1 + \mathcal{O} ( e^{ -2 \alpha_n / \maxmu{}_{\eps} } ) } } \\
&= \limf \lrhaa{ 1 - (\maxmu - \eps )^{-1} + O (\alpha_n^{-1}) } = 1-(\maxmu-\eps)^{-1}.
\end{align*}
Since $\eps$ was chosen arbitrarily, we obtain
\begin{align*}
\limf \logEn (\mu_n) \geq 1 - \frac 1{\maxmu }.
\end{align*}

We continue with the proof of the limsup inequality. We can restrict to $\xi \in \nsetcdeb$, because otherwise $\logE$ is infinite. We conclude the limsup inequality from Lemma \ref{lem diag arg} after showing that its two conditions are satisfied. We use Lemma \ref{lem diag arg} with the subset $\dssetlog := \bigaccv{ \xi \in \nsetcdeb }{ m_\xi > 0 }$.

\emph{Condition \eqref{lem diag arg cond i}.} Let $\xi \in \dssetlog$. We construct $\xi_n$ by using linear interpolation (see \eqref{for defn mun xin b}) with $x_i^n := \sup_{t < i/n} \xi(t)$ (because $\xi$ need not be in $W^{1,1}$, $\xi$ can not be evaluated at specific values). From Proposition \ref{prop ct conv EnF} we conclude that for $n$ large enough it holds
\begin{equation*}
    \EnF 5{(2-3)} (\xi_n) < \EF 5{(2-3)} (\xi) + 1 \leq \frac{\beta}2 + 1.
\end{equation*}

Observe that for any $i,j \in \{0, \ldots, n\}$ with $i > j$, we have the estimate
\begin{equation*}\label{thm log pf 1}
    \lrhaa{ x^n_i - x^n_j } = D \xi \bighaa{ [j/n, i/n) } \geq m_\xi \frac{i - j}n.
\end{equation*}
This is a similar estimate as \eqref{thm intro pf 5}. This allows us to use \eqref{thm intro pf 4} to derive the following upper bound
\begin{align*}
\Enita 5{(2-3)} (\xi_n) &= \frac{ \exp \bighaa{2 \bighaa{ \alpha_n - 1 } } }{ n \alpha_n } \sum_{k=1}^n \sum_{j = 0}^{n - k} V \bighaa{ n \alpha_n \bighaa{ x_{j+k}^n - x_j^n } } \\
&\leq \frac {2 m_\xi}{e^2} e^{2 \alpha_n (1 - m_\xi)} \bighaa{ 1 + \mathcal{O} \haa{ e^{-2 m_\xi \alpha_n} } } \rightarrow 0.
\end{align*}

By combining the estimates on $\Enita 5{(2-3)} (\xi_n)$ and $\EnF 5{(2-3)} (\xi_n)$ we obtain
\begin{align*}
    \limp \logEn (\xi_n)
&= \limp \frac 1{2 \alpha_n} \log \lrhaa{ \Enita 5{(2-3)} (\xi_n) + \EnF 5{(2-3)} (\xi_n) } \\
&\leq \limp \frac 1{2 \alpha_n} \log \lrhaa{ \frac {2 m_\xi}{e^2} e^{2 \alpha_n (1 - m_\xi)} \bighaa{ 1 + \mathcal{O} \haa{ e^{-2 m_\xi \alpha_n} } } + \frac{\beta}2 + 1 } \\
&= 1 - m_\xi = \logE (\xi).
\end{align*}

\emph{Condition \eqref{lem diag arg cond ii}.} Let $\xi \in \nsetcdeb$, and define
\begin{equation*}
  \xi_n(t) := \frac{ \xi(t) + t \eps_n }{ 1 + \eps_n }
\end{equation*}
for some $\eps_n \downarrow 0$. By construction, $\xi_n \in \dssetlog$, which follows from
\begin{gather*}
    \sup \xi_n = \frac{ \sup \xi + \eps_n}{1 + \eps_n} \leq 1, \hsf \text{and} \\
m_{\xi_n} = \inf_{b > a} \frac{ D \xi \bighaa{(a, b)} + \eps_n \mathcal{L} \bighaa{(a, b)} }{ b - a } = m_\xi + \eps_n > 0.
\end{gather*}
Clearly $\xi_n \weakto \xi$ in $BV_{\text{loc}}(0, 1)$, and
\begin{equation*}
\logE (\xi_n) = 1 - m_{\xi_n} = 1 - (m_\xi + \eps_n) \rightarrow \logE (\xi).
\end{equation*}
\end{proof}

\begin{acknowledgements}
We would like to thank M. Geers, R. Peerlings, M. Hütter, M. Kooiman and M. Dogge for fruitful discussions which have led us to our microscopic model, and L. Scardia for many discussions throughout this project.

PvM kindly acknowledges the financial support from the Netherlands Organisation for Scientific Research (NWO). He is part of the CorFlux project (nr. 10012310), which is one of the fourteen projects in the complexity program of NWO.
\end{acknowledgements}

\appendix

\section{Technical steps}

\begin{lem} \label{lem supp prop}
(A support property of narrow convergence). Let $\mu_n, \mu \in \mathcal{P} ([0, \infty))$ and ${a_n,a \in \R{}}$ such that $\mu_n \rightharpoonup \mu$ and $a_n \rightarrow a$. If $\supp \mu \nsubseteq [0, a]$, then $\supp \mu_n \nsubseteq [0, a_n]$ for all $n$ large enough.

\end{lem}

\begin{proof}[Proof of Lemma~\ref{lem supp prop}]
The proof goes by contradiction. Suppose there exists a subsequence $(\mu_n)$ such that $\supp \mu_n \subset [0, a_n]$. $\supp \mu \nsubseteq [0, a]$ and inner regularity imply that there is a closed interval $K$ in $(a, \infty)$ such that $\mu(K) > 0$. It is straightforward to choose a test function $\varphi \in C_b ( [0, \infty) )$ such that
\begin{equation*}
\intabx{0}{\infty}{\varphi}{\mu_n} = 0
\end{equation*}

\noindent for $n$ large enough, and
\begin{equation*}
\intabx{0}{\infty}{\varphi}{\mu} > 0,
\end{equation*}

\noindent which contradicts with $\mu_n \rightharpoonup \mu$.
\end{proof}

The following theorem is a generalization of [\nbcite{b32}, Theorem 4], in the sense that it applies to the sets $\nsetcdeblol$ and $\dssetdblol$ (see \eqref{thm intro pf 3}) not only for $\Lol = \infty$, but also for any $\Lol \in (0, \infty)$. The proof in \cite{b32} holds for finite $\Lol$ as well. % we kunnen de apx seq  g^\ell  namelijk zo kiezen dat  u_l \leq u  in de L^\infty norm.

\begin{thm} \label{thm T2.4}
(A sufficient condition for energy density). Let $f : (0, \infty) \rightarrow \bar{\R{}}$ be convex and decreasing, such that $\lim_{t \rightarrow \infty} f(t) = 0$. Let $E : \nsetcdeblol \rightarrow \bar{\R{}}$,
\begin{equation*}
    E(u) := \intabx{0}{1}{ f \lrhaa{ u'(t) } }{t}.
\end{equation*}
Then $\dssetdblol$ is energy dense in $\nsetcdeblol$ with respect to $E$.

\end{thm}

\begin{rem}
Just as in \cite{b32}, we use Theorem \ref{thm T2.4} for $\Eita{p}{q}$ for $p = 3,4,5$ and $q = 0, \ldots,3$. In these cases, we take for $f(t)$ respectively
\begin{equation*}
    \frac1t, \hsf \sum_{k = 1}^\infty V (k t), \hsf \chi_{ \acc{ t \geq 1 } }.
\end{equation*}

\end{rem}

\begin{lem} \label{lem ene dens 12}
(Energy density results).
Let $p \in \{1, 2\}$, and $\nsetabb$ and $\dssetabb$ as defined by \eqref{thm intro pf 3}. Then
\begin{enumerate}[(i)]
  \item $\dssetaba$ is energy dense in $\nsetab$ with respect to $E^{(p, q)}$ for $q = 0,1$. \label{lem ene dens 12 part i}
  \item $\dssetabb$ is energy dense in $\nsetabb$ with respect to $\Eita{p}{q}$ for $q = 2,3$. \label{lem ene dens 12 part ii}
\end{enumerate}

\end{lem}

\begin{proof}[Proof of Lemma~\ref{lem ene dens 12}]
Lemma \ref{lem ene dens 12}.\eqref{lem ene dens 12 part i} is proved by using Lemma \ref{lem diag arg alt cond ii}. It involves the intermediate space
\begin{equation*}
 \tilde{\mathcal{Y}} := \Bigaccv{ \mu \in \nsetab }{ \mu \ll \mathcal{L} \textrm{, } \ga{\mu}{\mathcal{L}} \in L^\infty (0, \infty) }.
\end{equation*}
The proof of $\tilde{\mathcal{Y}}$ being energy dense in $\nsetab$ is stated in [\nbcite{b32}, proof of limsup inequality Theorem 5]. The related sequence of the limsup inequality is given by the measures corresponding to the densities given by
\begin{equation} \label{for req seq cond i p 1 2}
    \rho_n(x) := n \mu \bighaa{ [x, x + 1/n) }.
\end{equation}

Still to be shown: \ddollar{ \forall \mu \in \tilde{\mathcal{Y}} \hsq \exists (\mu_n) \subset \dssetaba : \mu_n \rightharpoonup \mu, \textrm{ and } \limp E^{(p, 0)}(\mu_n) \leq E^{(p, 0)}(\mu). }

Let $\rho := \ga{\mu}{\mathcal{L}}$, and take $\rho_n = \frac1{ \mu([0, n]) } \rho \mathds{1}_{[0, n]}$. For the related $\mu_n$, it is easy to see that $\mu_n \rightharpoonup \mu$ and that $\EF{p}{0}(\mu_n) \leq \EF{p}{0}(\mu)$. By using the Dominated Convergence Theorem, one can prove $\Eita{p}{0}(\mu_n) \rightarrow \Eita{p}{0}(\mu)$.

The proof above works just as well for proving Lemma \ref{lem ene dens 12}.\eqref{lem ene dens 12 part ii}, because we can identify $\Eita{p}{(2-3)} |_{\nsetabb}$ by $\Eita{p}{0} |_{\nsetabb}$, since $\ahat_n$ and $\alpha_n$ play the same role.
\end{proof}
% See a.2.117

\section{$\Gamma$-convergence}

Here, we state the basic properties of $\Gamma$-convergence, which can be found, for example, in \cite{b16}. Although $\Gamma$-convergence can be defined on topological spaces, we only show the definition for metric spaces:

\begin{defn} \label{defn Gamma convergence}
($\Gamma$-convergence).
Let $(X, d)$ be a metric space, $E_n, E : X \rightarrow \overline{\R{}}$. Then~$E_n$ $\Gamma$-converges to $E$ with respect to $d$ iff the following two conditions are satisfied:
\begin{enumerate}[(i)]
  \item \ddollar{ \forall x \in X \hsq \forall x_n \xrightarrow{d} x : \limf E_n(x_n) \geq E (x) }, \label{defn Gamma convergence i}
  \item \ddollar{ \forall x \in X \hsq \exists y_n \xrightarrow{d} x : \limp E_n(y_n) \leq E (x) }. \label{defn Gamma convergence ii}
\end{enumerate}

\end{defn}

The sequence $(y_n)$, if it exists, is called the recovery sequence.

We continue by stating some properties of $\Gamma$-convergence that are useful to us. Let~$(X, d)$ be a metric space, and $E_n, E : X \rightarrow \overline{\R{}}$. The next Theorem [\nbcite{b16}, Proposition 6.20] states one of the most important properties of $\Gamma$-convergence. We need the following definition first:

\begin{defn} \label{defn continuous convergence}
(Continuous convergence).
Let $F_n, F : X \rightarrow \R{}$. Then $F_n \rightarrow F$ continuously iff
\begin{equation*}
    \forall x \in X : \lim_{\eps \rightarrow 0} \limp \sup_{y \in B(x, \eps )} \bigabs{ F_n(y) - F(x) } = 0.
\end{equation*}

\end{defn}

\begin{thm} \label{thm gamma conv stable cont pert}
(Stability of $\Gamma$-convergence under continuously converging perturbations).
Let $F_n, F : X \rightarrow \R{}$. If $F_n \rightarrow F$ continuously, then
\begin{equation*}
    E_n + F_n \xrightarrow{\Gamma} E + F.
\end{equation*}

\end{thm}

\begin{rem}
Note that $R(F) \subset \R{} \subsetneq \overline{\R{}}$ is required in Definition \ref{defn continuous convergence}.

\end{rem}

\end{document}